%% file: POA-POS-Tradeoff.tex
\documentclass[10pt,twocolumn,twoside]{IEEEtran} 
\newtheorem{theorem}{Theorem}[section]

\newtheorem{proposition}{Proposition}[section]
\newtheorem{lemma}{Lemma}[section]

\newcommand{\arr}{\mathbb{R}}

\newcommand{\aee}{{\cal A}}
\newcommand{\ree}{{\cal R}}
\newcommand{\gee}{{\cal G}}
\newcommand{\nee}{{\cal N}}

\newcommand{\qed}{\ \hfill{$\Box$}\\[5pt]}

\newcommand{\DP}[1]{{#1}}
\newcommand{\comment}[1]{}
\newcommand{\pad}{\vspace{.07in}}

\newcommand{\fes}{f^{\rm es}}
\newcommand{\fmc}{f^{\rm mc}}
\newcommand{\fgair}{f^\star}
\newcommand{\fsb}{f^{\rm sb}}
\newcommand{\ane}{a^{\rm ne}}
\newcommand{\aopt}{a^{\rm opt}}
\newcommand{\aoptprime}{a^{\rm opt'}}
\newcommand{\card}[1]{|#1|}

\usepackage{cite}

\renewcommand{\baselinestretch}{.965}

\usepackage{enumerate}
\usepackage{cite}
\usepackage{latexsym}
\usepackage{amsmath}
\usepackage{amsfonts}
\usepackage{amssymb}
\usepackage{graphicx}
\usepackage{times}
\usepackage{sgame}
\usepackage{color}

\usepackage{tikz}
\usepackage{pgfplots}
\graphicspath{{figures/}}

\IEEEoverridecommandlockouts                              %
\usepackage{fancyhdr,amssymb,amsmath, graphicx, listings,float,subfig,enumerate,epstopdf,color,setspace,bm,cite,setspace,cite}

\normalsize\title{{Multiagent {Maximum} Coverage Problems: \\ The Trade-off Between Anarchy and Stability}
\thanks{This work is supported by SNSF Grant \#P2EZP2-181618, ONR Grant \#N00014-15-1-2762 and NSF Grant \#ECCS-1351866. The conference version of this work is given in \cite{ram2019}.}}

\author{
Vinod Ramaswamy, Dario Paccagnan, and Jason R. Marden  
\thanks{V. Ramaswamy was with the Department of Electrical, Computer, and Energy Engineering, University of Colorado at Boulder, \texttt{mailtovinodr@gmail.com}.}
\thanks{
D. Paccagnan is with the Department of Mechanical Engineering and with the Center of Control, Dynamical Systems and Computation, University of California, Santa Barbara, \texttt{dariop@ucsb.edu}.}
\thanks{J. R. Marden is with the Department of Electrical and Computer Engineering, University of California, Santa Barbara, \texttt{jrmarden@ece.ucsb.edu}.}
}

\begin{document}

\maketitle

\begin{abstract}
\DP{The price of anarchy and price of stability are three well-studied performance metrics that seek to characterize the inefficiency of equilibria in distributed systems. The distinction between these two performance metrics centers on the equilibria that they focus on: the price of anarchy characterizes the quality of the worst-performing equilibria, while the price of stability characterizes the quality of the best-performing equilibria. 
While much of the literature focuses on these metrics from an analysis perspective, in this work we consider these performance metrics from a design perspective.  Specifically, we focus on the setting where a system operator is tasked with designing local utility functions to optimize these performance metrics in a class of games termed covering games. 
Our main result characterizes a fundamental trade-off between the price of anarchy and price of stability in the form of a fully explicit Pareto frontier. Within this setup, optimizing the price of anarchy comes directly at the expense of the price of stability (and vice versa). 
Our second results demonstrates how a system-operator could  incorporate an additional piece of system-level information into the design of the agents' utility functions to breach these limitations and improve the system's performance. This valuable piece of system-level information pertains to the performance of worst performing agent in the system.}
\end{abstract}

{\color{red}
}
\section{Introduction}
A multiagent system can be characterized by a collection of individual subsystems, each making independent decisions in response to locally available information. Such a decision-making architecture can either emerge naturally as the results of self-interested behavior, e.g., drivers in a transportation network, or be the result of a design choice in engineered system.  In the latter case, the need for distributed decision-making stems from the scale, spatial distribution, and sheer quantity of information associated with various problem domains that exclude the possibility for centralized decision making and control. One concrete example is the problem of monitoring the perimeter of a wild fire, where the goal is to deploy a collection of unmanned aerial vehicles (UAV) to effectively survey the perimeter of a wild fire under the constraint that each UAV makes independent decisions in response to local information regarding its own aerial view and minimal information regarding the state of neighboring UAVs \cite{Alexis2009}.  Alternative examples include the use of robotic networks in post-disaster environments \cite{kuntze2012seneka, kitano1999robocup}, task scheduling and management \cite{ernst2004staff}, water conservative food production \cite{kozai2015plant}, fleets of autonomous vehicles \cite{spieser2014toward}, and micro-scale medical treatments  \cite{servant2015controlled,ishiyama2002magnetic}.

Regardless of the specific problem domain, the central goal in the design of a networked control system is to derive \emph{admissible} control policies for the decision-making entities that ensure the emergent collective behavior is desirable with regards to a given system-level objective. At a high level, this design process entails specifying two key elements: the information available to each subsystem, attained either through sensing or communication, and a decision-making mechanism that prescribes how each subsystem processes available information to take decisions. The quality of a networked control architecture is ultimately gauged by several dimensions including the stability and efficacy  of the emergent collective behavior, characteristics of the transient behavior, in addition to communication costs associated with propagating information throughout the system. \DP{Within this setting, two fundamental questions arise:
\begin{enumerate}
	\item[(i)] What are the decision-making rules that optimize the performance of the emergent collective behavior?%
	\vspace*{1mm}
	\item[(ii)] How does informational availability translate to attainable performance guarantees?%
\end{enumerate} }

This paper seeks to shed light on the answer to these two questions for the class of multiagent maximum coverage problems introduced in \cite{Gairing09}. \DP{
In a multiagent maximum coverage problem we are given a ground set of resources $\ree$, $n$ collections of subsets of the ground set $\aee_1, \dots, \aee_n$ where $\aee_i \subseteq 2^\ree$ for each $i \in \{1,\dots, n\}$, and a valuation $v_r \geq 0$ for each resource $r \in \ree$.  Given the specific covering problem, the system-level objective is to select one set  from each collection, i.e., $a_i \in \aee_i$ for each $i=1, \dots, n$,  so as to maximize the total value of covered elements. It is important to highlight that there are well-established centralized algorithms that can derive an admissible allocation of agents to resources that is within a factor of $1-1/e$ of the optimal allocation's value in polynomial time, provided that the $n$ collections of subsets coincide (i.e. $\aee_i = \aee_j$ for all $i, j$) and further technical assumption hold \cite{nem78-I,4031401,krause07}.} Further, no polynomial time algorithm can provide a better approximation, unless $\mathcal{P}=\mathcal{NP}$. Unfortunately, the applicability of such centralized algorithms for the control of multiagent systems is limited given the concerns highlighted above. 

\DP{This paper focuses on distributed approaches for reaching a near-optimal allocation where the individual agents make their covering selections in response to a \emph{designed} decision-making policy that makes use of locally available information.} The central goal here is to design agent decision-making rules that optimize the quality of the emergent collective behavior for a given level of informational availability. Of specific interest will be identifying how the level of information available to the individual agents impacts the attainable performance guarantees associated with the corresponding optimal networked control system. 

In the spirit of \cite{Gairing09,MardenWierman13}, we approach this problem through a game theoretic lens where we model the individual agents as players in a game and each agent is associated with a local objective function that guides its decision-making process. We treat these local objective functions as our design parameter and focus our analysis on characterizing the performance guarantees associated with the resulting equilibria of the designed game. Here, we model the informational restrictions discussed above as limitions on the amount of information that these local agent objective functions can depend on.  We concentrate our analysis on two well-studied performance metrics in the game theoretic literature termed the \emph{price of anarchy} and \emph{price of stability} \cite{Anshelevich04,Koutsoupias1999}. Informally, the price of anarchy provides performance guarantees associated with  the worst performing equilibrium relative to the optimal allocation. The price of stability, on the other hand, provides similar performance guarantees when restricting attention to the best performing equilibrium. The lack of uniqueness of equilibria implies that these bounds are often quite different.\footnote{The justification for focusing purely on equilibria, as opposed to dynamics, derives from the fact that there is a rich body of literature in distributed learning that, coupled with the derived objective, lead the collective behavior to an equilibrium, c.f., \cite{fudlev}.}

The work of \cite{Gairing09} was one of the first to view price of anarchy as a design objective rather than its more traditional analytical counterpart.  The main results in \cite{Gairing09} identify a set of agent objective functions that optimize the price of anarchy when agents are only aware of (i) the resources the agent can select and (ii) the number of agents covering these resources.  Note that in this setting, any agent $i$ is unaware of the covering options of any other agents $j\neq i$, as well as any resource values that the agent is unable to cover.  Interestingly, \cite{Gairing09} demonstrates that this optimal price of anarchy attains the same $1-1/e$ guarantees of the best centralized algorithms (even when the $n$ collections of subsets are different), meaning that there is no degradation in terms of the worst-case efficiency guarantees when transitioning from the best centralized algorithm to the presented distributed algorithm that adheres to the prescribed informational limitations.

\pad
{\noindent \bf Our Contribution.}
\DP{The first main results of this manuscript addresses the achievable performance guarantees, in terms of the price of anarchy and price of stability, attainable through the design of agent objective functions of the above form in \cite{Gairing09}.    }\DP{In Theorem~\ref{thm:tradeoff}, we characterize the price of anarchy and price of stability frontier that is achievable through the design of agent objective functions in these multiagent covering problem.  This characterization demonstrates a fundamental trade-off between the price of anarchy and price of stability as design objectives in such multiagent covering problems.}  That is, designing agent objective functions to improve the worst-case performance guarantees necessarily degrades the best-case performance guarantees.  As corner cases, we demonstrate that any objective functions that ensure a price of anarchy of $1-1/e$ also inherit a price of stability of $1-1/e$.  Note that having a price of stability smaller than $1$ implies that the optimal allocation is not necessarily an equilibrium.  Alternatively, any objective functions that ensure a price of stability of $1$ also inherit a price of anarchy of at most $1/2$. %

\DP{The second main result of this manuscript demonstrates that one can move beyond this frontier by providing through the design of agent objective functions encompassing additional system-wide information.  } In Theorem~\ref{t4}, we identify a minimal (and easily attainable) piece of system-level information that permits the realization of decision-making rules with performance guarantees beyond the price of anarchy / price of stability frontier provided in Theorem~\ref{thm:tradeoff}. When agents are provided with this additional information, which can be informally interpreted as the largest value of an uncovered resource in the system, one can derive agent objective functions that yield a price of anarchy of $1-1/e$ and a price of stability of $1$, which was unattainable without this piece of information. 

\DP{The importance of this result centers on the fact that specific system-level information, if propagated to the agents, could be exploited in networked control algorithms to improve system performance. Understanding the tradeoff between this improvement and the communication costs necessary to propagate this information throughout the system is clearly an important question that warrants future attention.}

\pad
{\noindent \bf Related Work.}
The results contained in this manuscript add to the growing literature of utility design, which can be interpreted as a subfield of mechanism design \cite{christodoulou2004coordination} where the objective is to design admissible agent objective functions to optimize various performance metrics, such as the price of anarchy and price of stability,  \cite{marden-connections,marsha2013,Korman2018arxiv}. 
While recent work in \cite{Ragavendran14} has identified all design approaches that ensure equilibrium existence in local utility designs, the question of optimizing the worst-case efficiency of the resulting equilibria, i.e., optimizing the price of anarchy, is far less understood. Nonetheless, there are a few positive results in this domain worth reviewing. 
Beyond \cite{Gairing09}, alternative problem domains where optimizing the price of anarchy has been explored include concave cost sharing games \cite{Marden17}, optimal tolling in congestion games \cite{paccagnan2019incentivizing} and reverse carpooling games \cite{Marden12}. More recently, the authors in \cite{paccagnan2018arxiv1,paccagnan2018arxiv2} characterize and optimize the price of anarchy relative to a broader class of submodular and supermodular combinatorial optimization problems, rediscovering \cite{Gairing09} as a special case. Lastly, a recent result in \cite{Lazos2018arxiv} characterizes a similar trade-off between the price of anarchy and price of \mbox{stability in a mechanism design setting.} 

\DP{Much of the research regarding optimal utility design has concentrated on a specific class of objectives, termed \emph{budget-balanced} objectives, which imposes the constraint that the sum of the agents' objectives is equal to the system welfare for every allocation. Within the confines of budget-balanced agent objective functions, several works have identified the optimality of the Shapley value objective design with regards to the price of anarchy guarantees \cite{optCostSharing,facilityLocationGames,chen}.  However, the imposition of the budget-balanced constraint is unwarranted in the context of multiagent system design and its removal allows for improved performance, \mbox{as shown in \cite{Gairing09} and this manuscript.}}

\pad
\DP{{\noindent \bf Organization.} In section \ref{sec:modelmetrics} we introduce the multiagent coverage problem and corresponding performance metrics. In section \ref{sec:prelim} we provide some preliminary analysis on the performance of specific distribution rules. In section \ref{sec:tradeoff} we present the trade-off between price of anarchy and stability. Lastly, in section \ref{sec:additionalinfo} we show how to breach this trade-off by leveraging an additional piece of system-level information.}

\pad
{\noindent \bf Notation.} 
We use $\mathbb{N}$, $\mathbb{R}_{>0}$ and $\mathbb{R}_{\ge0}$ to denote the set of natural, positive and non negative real numbers; $e$ is Euler's number.

\section{Model and Performance Metrics}
\label{sec:modelmetrics}
In this section we introduce the multiagent maximum coverage problem and our game theoretic model for the design of local decision-making mechanisms \cite{Gairing09}.  Further, we define the objectives and performance metrics of interest, as well as provide a review of the relevant literature.  
\subsection{Covering problems}

Let $\ree = \{r_1, r_2, \dots, r_m\}$ be a finite set of resources where each resource $r \in \ree$ is associated with a value $v_r \geq 0$ defining its importance.  
We consider a covering problem where the goal is to allocate a collection of agents $N = \{1, \dots, n\}$ to resources in $\ree$ in order to maximize the cumulative value of the covered resources. The set of possible assignments for each agent $i\in N$ is given by $\aee_i \subseteq 2^\ree$ and we express an allocation by the tuple $a=(a_1, a_2, \dots, a_n) \in \aee=\aee_1 \times \dots \times \aee_n$. \DP{The agents' assignment sets $\aee_1, \dots, \aee_n$ are fixed, and need not conform to any specific structure or symmetry.} The total value, or welfare, associated with an allocation \mbox{$a \in \aee$ is given by}
\begin{equation}
W(a)= \sum_{r \in \cup_{i \in N} a_i} v_r.
\end{equation}
\DP{The goal of the covering problem is to find an optimal allocation, i.e., an   allocation $a^{\rm opt}\in \aee$ that satisfies
	\begin{equation}
	a^{\rm opt} \in \underset{a \in \aee}{\arg \max} \ W(a),
	\end{equation} 
where we restrict attention to admissible allocations $a \in \aee$.  
We will express an allocation $a$ as $(a_i, a_{-i})$ with the understanding that $a_{-i}=(a_1, \dots, a_{i-1}, a_{i+1}, \dots, a_n)$ denotes the collection of choices of the agents other than agent~$i$.  Lastly, we will periodically restrict our attention to the special case of single-selection covering problems where each allocation $a_i \in \aee_i$ consists of a single resource.  Figure~\ref{fig:examples} provides some illustrative examples of single-selection covering problems.  }

\begin{figure}[t] 
	\begin{center}
		\includegraphics[width=0.99\columnwidth]{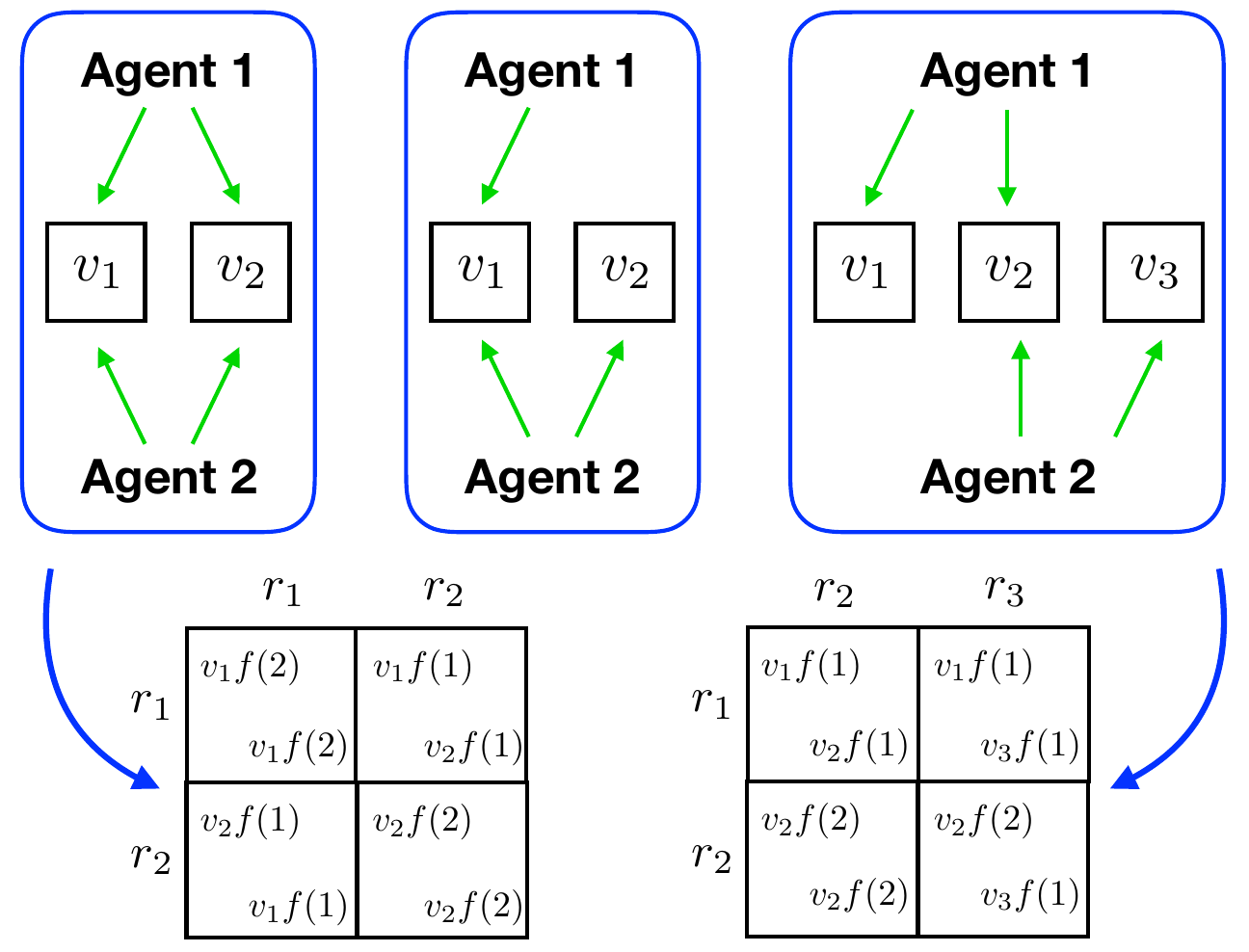}
		\caption{ \DP{The figure provides an illustration of three different single-selection covering problems with $N=\{1,2\}$. For the covering problem on the left there are two resources with valuations $v_1, v_2$. Further, the action sets for both agents are symmetric and of the form $\aee_1 =\aee_2 = \{\{r_1\}, \{r_2\}\}$; hence, either agent can select either resource. 
			The covering problem in the middle is identical to that on the left with the exception being that the allocation choice of agent $1$ is now of the form $\aee_1 = \{\{r_1\}\}$, meaning that agent $1$ is now only capable of choosing resource $r_1$.  
			For the covering problem on the right there are three resources with valuations $v_1, v_2, v_3$. Further, the action sets for both agents are not symmetric and of the form $\aee_1 =\{\{r_1\}, \{r_2\}\}$ and $\aee_2 =\{\{r_2\}, \{r_3\}\}$; hence, agent $1$ is not able to select resource $r_3$ and similarly agent $2$ is not able to select resource $r_1$.  It is important to highlight that the structure of the agents' choice sets is given, i.e., not a design choice, and the goal is to find the optimal allocation within the given admissible choice set. In general, the allocation choices of an agent need not be restricted to singletons as highlighted above. 
			Lastly, the two payoff matrices highlight the utility function given in \eqref{eq:ut} of both agents for the left and right scenarios respectively. Agent $1$ is the Row player, Agent $2$ is the Column player, and the numbers in each quadrant correspond to their payoffs for the given joint action.  Here, the top left entry is agent 1's payoff and the bottom right entry is agent 2's payoff.  Note that the utilities are expressed with regards to the distribution rule $f$. }  }
		\label{fig:examples}
	\end{center}
\end{figure}

\subsection{A game theoretic model}

This paper focuses on deriving distributed mechanisms for attaining near optimal solutions to covering problem where the individual agents make independent choices in response to local available information. Specifically, in this section we assume that each agent $i$ has information only regarding the resources that the agent can select. Rather than directly specifying a decision-making process, here we focus on the design of local agent objective functions that adhere to these informational dependencies and will ultimately be used to guide the agents' selection process. To that end, we consider the framework proposed in \cite{Gairing09} where each agent is associated with a local utility or objective function $U_i : \mathcal{A} \rightarrow \arr$, and for any allocation $a = (a_i, a_{-i}) \in \aee$, the utility of agent $i$ is 
\begin{equation}\label{eq:ut}
U_i(a_i, a_{-i}) = \sum_{r \in a_i} v_r \cdot f(|a|_r),
\end{equation} 
\DP{where $|a|_r$ captures the number of agents that choose resource $r$ in the allocation $a$, i.e., the cardinality of the set $\{ i\in N: a_i=r\}$, and $f: \mathbb{N} \rightarrow \arr$ defines the fractional benefit of the resource valuation awarded to each agent for selecting a given resource in an allocation $a$.  In order to evaluate this utility, it is important to highlight that each agent $i$ is required to observe (i) the resources in its action set and (ii) the number of other agents covering each of these resources. We will refer to $f$ as the \emph{distribution rule} throughout.}\footnote{\DP{The structure of utility functions considered in \eqref{eq:ut} is common in the literature and falls under the framework of networked cost-sharing games or distributed welfare games, c.f., \cite{MardenWierman13}. This paper focuses on characterizing the optimal performance guarantees attainable with this utility structure.  Once such guarantees are known, one can evaluate whether or not it is desirable to raise or lower the informational dependence on the utility structure.}}  We will express such a $n$-agent welfare sharing game by the tuple $G^n=\left\{N, \ree, \{\aee_i\}_{i\in N}, f, \{v_r\}_{r \in \ree} \right\}$ and drop the subscripts on the above sets, e.g., denote $\{v_r\}_{r \in \ree}$ as $ \{v_r\}$, for brevity.

The goal of this paper is to derive the distribution rule $f$ that optimizes the performance of the emergent collective behavior.  Here, we focus on the concept of pure Nash equilibrium as a model for this emergent collective behavior \cite{Nash50}.  A pure Nash equilibrium, which we will henceforth refer to as just an equilibrium, is defined as an allocation $a^{\rm ne} \in \aee$ such that for all $i \in N$ and for all $a_i\in\aee_i$, we have %
\[
U_i(a_i^{\rm ne}, a_{-i}^{\rm ne}) \ge U_i(a_i, a_{-i}^{\rm ne}).
\]
In essence, an equilibrium represents an allocation for which no single agent has a unilateral incentive to alter its covering choice given the choices of the other agents.
\DP{
It is important to highlight that an equilibrium might not exist in a general game; however, when restricting attention to a covering game as defined above, an equilibrium is guaranteed to exist as the resulting game is known to be a congestion game \cite{Monderer96}.}\footnote{There is a rich body of literature that provides distributed algorithms that can coordinate the agents to an equilibrium for the class of games considered in this paper \cite{you,you2,fudlev,marden05-jsfp-journal,marsha2013}. However, we will not discuss such results due to space considerations.}  

We will measure the efficiency of an equilibrium allocation in a game $G^n$ through two commonly studied measures, termed \emph{price of anarchy} and \emph{price of stability}, defined as
\[
\begin{split}
{\rm PoA}(G^n) &= \min_{a^{\rm ne} \in {\rm NE}(G^n)} \frac{W(a^{\rm ne})}{W(a^{\rm opt})}\le 1, \\[2mm]
{\rm PoS}(G^n) &= \max_{a^{\rm ne} \in {\rm NE}(G^n)}  \frac{W(a^{\rm ne})}{W(a^{\rm opt})}\le 1, 
\end{split}
\]
where we use ${\rm NE}(G^n)$ to denote the set of Nash equilibria of $G^n$. In words, the price of anarchy characterizes the performance of the worst equilibrium of $G^n$ relative to the performance of the optimal allocation, while the price of stability focuses on the best equilibrium in the game $G^n$. Such distinction is required as equilibria are guaranteed to exists for the class of utilities considered in \eqref{eq:ut}, but in general they are not unique.
By definition $0\le {\rm PoA}(G^n) \le {\rm PoS}(G^n) \le 1 $.

\DP{Throughout, we require that a system designer commits to a distribution rule \emph{without} explicit knowledge of particular covering problem.  To that end, the system operator will be aware of the maximum number of agents in the game $n$, but will have no specific knowledge of actual number of agents, the resource set $\ree$ and valuations $\{v_r\}$, as well as the agents' action sets $\{\aee_i\}$. For example, all of the three examples highlighted in Figure~\ref{fig:examples} represent instances of covering problems with at most two agents, and when designing $f$ a system designer has no a priori knowledge about which specific instance will be realized.  Note that once a particular distribution rule $f$ has been chosen, this distribution rule defines a game for \emph{any} realization of the parameters. The objective of the system designer is to provide desirable performance guarantees irrespective of the realization of these parameter, even if they where chosen by an adversary. To that end, let $\gee_f^n$ denote the family of covering agent games with at most $n$ agents that utilize a given distribution rule $f$, i.e., any game $G \in \gee_f^n$ is of the above form.}  We will measure the quality of a distribution rule $f$ by a worst-case analysis over the set of induced games $\gee_f^n$, which is the natural extension of the price of anarchy and price of stability defined above, i.e., 
\begin{eqnarray}\label{eq:poa2}
{\rm PoA}(\gee_f^n) &=& \min_{G^n \in \gee_f^n}  {\rm PoA}(G^n),  \\[2mm]
{\rm PoS}(\gee_f^n) &=& \min_{G^n \in \gee_f^n}  {\rm PoS}(G^n).  \label{eq:pos3}
\end{eqnarray}
The price of anarchy ${\rm PoA}(\gee_f^n)$ for a given distribution rule $f$ provides a bound on the quality of any equilibrium irrespective of the agent set $N$, resource set $\ree$, action sets $\{\aee_i\}$, and resource valuations $\{v_r\}$. The price of stability, on the other hand, provides similar performance guarantees when restricting attention to the best equilibrium.\footnote{One motivation for studying the price of stability is the availability of distributed learning rules that lead the collective behavior to the best equilibrium, e.g., \cite{marden-shamma-08,blume93,blume97}.}

\DP{
\section{Preliminary Results}\label{sec:prelim}
In this section we provide a preliminary study on the price of anarchy and price of stability for three well-studied distribution rules. The first distribution rule that we focus on is the \emph{equal share distribution rule}, which is defined as 
\begin{equation}
\label{eq:fes}
\fes(j)=\frac{1}{j}, \quad \forall j \geq 1.
\end{equation}
The following proposition provides the price of anarchy and price of stability guarantees associated with the equal share distribution rule.\\
\begin{proposition}\label{claim:eq-share}
	Consider the class of maximum coverage games $\gee_{\fes}^n$ with the equal share distribution rule $\fes$ and at most $n$ agents.  The price of anarchy and price of stability over the set of games $\gee_{\fes}^n$ satisfies
	\begin{eqnarray*}%
		\lim_{n \rightarrow \infty} \ {\rm PoS}(\gee_{\fes}^n) & = & 1/2, \\
		\lim_{n \rightarrow \infty} \ {\rm PoA}(\gee_{\fes}^n) & = & 1/2.
	\end{eqnarray*}
\end{proposition}
}

\DP{Before delving into the proof of Proposition~\ref{claim:eq-share}, we begin by presenting a lemma for characterizing the price of anarchy for any distribution rule
$f$ in the class of distribution rules $\mathcal{F}^n=\left\{f\in\mathbb{R}^n \right\}$.  This result strengthens the results presented in \cite{paccagnan2018arxiv2,Gairing09},  which focus purely on set of positive distribution rules that satisfies the constraint $f(1)=1$, i.e., $\mathcal{F}^n_1=\left\{f\in\mathbb{R}^n_{\ge0} : f(1)=1\right\}.$   }

\DP{
\begin{lemma}\label{t1}
	 Let $n \geq 2$. The following properties hold: \\ (i)  The price of anarchy associated with any distribution rule $f \in {\cal F}^n_1$ is
		\begin{equation}\label{eq:poa-guarantee}
		{\rm PoA}(\gee_f^n) = \frac{1}{1+\chi_f^n},
		\end{equation}
		where 
		\begin{equation}
		\label{eq:chi}
		\chi_f^n  \hspace{-0.05cm}= \hspace{-.15cm}\max_{j\le n-1 } \hspace{-.1cm} \left\{(j+1)f(j+1)-1, jf(j)-f(j+1), jf(j+1) \right\}.
		\end{equation}
	\\ (ii) If $f(1) \leq 0$ or $f(k) < 0$ for any $k\geq 2$, the price of anarchy satisfies ${\rm PoA}(\gee_f^n) = 0$. \\
	(iii) The optimal price of anarchy over the set of distribution rules ${\cal F}^n$ is given by
	\begin{equation}\label{eq:opt-poa}
	\max_{f \in {\cal F}^n}{{\rm PoA}(\gee_f^n)} = 1 - \frac{1}{\frac{1}{(n-1)(n-1)!} + \sum_{i=0}^{n-1} \frac{1}{i!}},
	\end{equation}
	which satisfies 
	\begin{equation}\label{eq:opt-poa3}
	\lim_{n \rightarrow \infty} \max_{f \in {\cal F}^n}{{\rm PoA}(\gee_f^n)} = 1 - 1/e.
	\end{equation}
	\\ (iv) Lastly, a distribution rule $f^\star \in {\cal F}^n$ that achieves the optimal price of anarchy in \eqref{eq:opt-poa} is unique and of the form
	\begin{equation}\label{eq:fgair}
	\fgair(j) = (j-1)! \left(\frac{ \frac{1}{(n-1)(n-1)!} + \sum_{i=j}^{n-1} \frac{1}{i!}}{\frac{1}{(n-1)(n-1)!} + \sum_{i=1}^{n-1} \frac{1}{i!}} \right), \ j\leq n.
	\end{equation}
\end{lemma}}

\DP{We will refer to the optimal distribution rule given in \eqref{eq:fgair} as the Gairing distribution rule. %
}

\DP{
\emph{Proof of Lemma \ref{t1}}
	 Part (i) merely restates the results given \cite{paccagnan2018arxiv2,Gairing09}, c.f., Theorem~2 in \cite {paccagnan2018arxiv2}. We will prove Part (ii) of this lemma by constructing a specific game instance where the price of anarchy is $0$, hence implying that the price of anarchy over the class of games $\gee_f^n$ is also $0$. Suppose $f(1)=0$ and consider a single-selection maximum coverage problem with three resources $\{r_1, r_2, r_3\}$ and valuations $v_1$, $v_2$, and $v_3 = 0$.  Furthermore, suppose $\aee_1 = \{\{r_1\}, \{r_2\}\}$ and $\aee_j = \{\{r_3\}\}$ for all agents $j \in \{2, \dots, n\}$. For any resource valuations $v_1 > v_2$, the action profile $(r_2, r_3, \dots, r_3)$ is an equilibrium while the action profile $(r_1, r_3, \dots, r_3)$ is the optimal allocation.  Since $v_1$ can be arbitrarily large relative to $v_2$, this completes the proof.  A similar class of examples can be constructed for the case where $f(k) < 0$ for some $k \geq 2$. For Part (iii), consider the family of distribution rules $\mathcal{F}^n_+=\left\{f\in\mathbb{R}^n_{\ge0} : f(1)>0\right\}.$ Building upon Part (ii), we know that   
	 \begin{equation}
	 \max_{f \in {\cal F}^n}{{\rm PoA}(\gee_f^n)} = \max_{f \in {\cal F}^n_+}{{\rm PoA}(\gee_f^n)}  = \max_{f \in {\cal F}^n_1}{{\rm PoA}(\gee_f^n)},
	 \end{equation}
	 where the second equality comes from the fact that the price of anarchy is invariant to scaling, i.e., ${\rm PoA}(\gee_f^n) = {\rm PoA}(\gee_{\alpha \cdot f}^n)$ for any $f \in {\cal F}^n$ and $\alpha > 0$. Part (iv) follows from the results given \cite{paccagnan2018arxiv2,Gairing09} coupled with the results in Part~(iii), i.e., there is no loss in restricting attention to distribution rules to ${\cal F}^n_1$. $\hfill \Box$
}

\DP{We will now proceed with the proof of Proposition~\ref{claim:eq-share}.}

\DP{\emph{Proof of Proposition~\ref{claim:eq-share}:} We will start with the price of anarchy result. First,  note that $\fes \in \mathcal{F}^n_1$. Using Lemma~\ref{t1}, the price of anarchy is ${\rm PoA}(\gee_{\fes}^n) = \frac{1}{1+\chi_{\fes}^n}$ where  
\begin{equation}
\label{eq:chi2}
\chi_{\fes}^n  \hspace{-0.05cm}= \hspace{-.15cm}\max_{j\le n-1 } \hspace{-.1cm} \left\{0, 1-\frac{1}{j+1}, \frac{j}{j+1} \right\} = \frac{n-1}{n}.
\end{equation}
Accordingly, we have that 
\begin{eqnarray*}%
	\lim_{n \rightarrow \infty} \ {\rm PoA}(\gee_{\fes}^n) = 
	\lim_{n \rightarrow \infty} \ \frac{1}{1+\frac{n-1}{n}} = 1/2.
\end{eqnarray*}
Moving to the price of stability, we know that ${\rm PoS}(\gee_{\fes}^n) \geq {\rm PoA}(\gee_{\fes}^n)$ by definition.  We will now construct a specific example with a price of stability equal to the price of anarchy which provides the equality.  To that end, consider a game with $n$ agents and $n+1$ resources denoted by $\{r_1, r_2, \dots, r_{n+1}\}$, where the action set of each agent $i \in N$ is $\aee_i = \{ \{r_i\}, \{r_{n+1}\}\}$.  Further, suppose the valuations of the resources satisfies $v_1 = \dots = v_n = 1$ and $v_{n+1}=n$.  If is straightforward to verify that the unique equilibrium is when all agents choose resource $r_{n+1}$, which yields a welfare of $W(a^{\rm ne})=n$.  An optimal allocation is each agent $i \in \{1,\dots, n-1
\}$ select resource $r_i$ while agent $n$ chooses resource $r_{n+1}$.  The welfare associated with the optimal allocation satisfies $W(a^{\rm opt})=n+(n-1)$.  Hence, the price of stability over the class of games must satisfy
\begin{eqnarray*}%
	\lim_{n \rightarrow \infty} \ {\rm PoS}(\gee_{\fes}^n) \leq 
	\lim_{n \rightarrow \infty} \ \frac{n}{n+{n-1}} = 1/2,
\end{eqnarray*}
which gives us the result since ${\rm PoA}(\gee_{\fes}^n)\le {\rm PoS}(\gee_{\fes}^n)$.  Note that the above analysis also ensures that ${\rm PoS}(\gee_{\fes}^n) = {\rm PoA}(\gee_{\fes}^n) = \frac{n}{2n-1}$ for any $n$.  $\hfill \Box$}\\

\DP{The second distribution rule that we focus on is the \emph{marginal contribution distribution rule}
\begin{equation}
\label{eq:fmc}
\fmc(j)=\left\{ \begin{array}{lcl} 1 & & \text{if} \ j=1, \\ 0 & & \text{otherwise} \end{array} \right.
\end{equation}
The following proposition provides the price of anarchy and price of stability guarantees associated with the marginal contribution distribution rule.\\
\begin{proposition}\label{claim:mc-share}
	Consider the class of maximum coverage games $\gee_{\fmc}^n$ with the marginal cost distribution rule $\fmc$ and at most $n$ agents.  The price of anarchy and price of stability over the set of games $\gee_{\fmc}^n$ satisfies
	\begin{eqnarray*}%
		\lim_{n \rightarrow \infty} \ {\rm PoS}(\gee_{\fmc}^n) & = & 1, \\
		\lim_{n \rightarrow \infty} \ {\rm PoA}(\gee_{\fmc}^n) & = & 1/2.
	\end{eqnarray*}
\end{proposition}}

\DP{\emph{Proof of Proposition~\ref{claim:mc-share}:} We will start with the price of anarchy result. Using Lemma~\ref{t1}, the price of anarchy is ${\rm PoA}(\gee_{\fmc}^n) = \frac{1}{1+\chi_{\fmc}^n}$ where $\chi_{\fmc}^n=1$, i.e., the maximum of \eqref{eq:chi} is attained for $j=1$.  Accordingly, the price of anarchy is ${\rm PoA}(\gee_{\fmc}^n) = 1/2$.  With regards to the price of stability, it is well-known that the price of stability satisfies ${\rm PoS}(\gee_{\fmc}^n) = 1$. This follows from \cite{marsha2013} as the resulting game is a potential game with potential function coinciding with the system welfare $W$. The fact that the potential function coincides with $W$ can be easily seen upon substituting the expression for $\fmc$ in the potential in \eqref{eq:potential}.  This completes the proof.$\hfill \Box$
}

\DP{The last distribution rule that we focus on is the Gairing distribution rule given in \eqref{eq:fgair}.
	\begin{proposition}\label{claim:gair-share}
		Consider the class of maximum coverage games $\gee_{\fgair}^n$ with the Gairing distribution rule $\fgair$ defined in \eqref{eq:fgair} and at most $n$ agents.  The price of anarchy and price of stability over the set of games $\gee_{\fgair}^n$ satisfies
		\begin{eqnarray*}%
			\lim_{n \rightarrow \infty} \ {\rm PoS}(\gee_{\fgair}^n) & = & 1-1/e, \\
			\lim_{n \rightarrow \infty} \ {\rm PoA}(\gee_{\fgair}^n) & = & 1-1/e.
		\end{eqnarray*}
\end{proposition}}

\DP{\emph{Proof of Proposition~\ref{claim:gair-share}:} The price of anarchy result comes directly from Lemma~\ref{t1}. With regards to the price of stability result, consider the game instance provided in the proof of Proposition~\ref{claim:eq-share} where $v_1 = \dots = v_n = 1$ and $v_{n+1}=1/\fgair(n)$. The structure of the unique equilibrium and optimal allocation are identical to those presented above, and the resulting welfare satisfies $W(\ane)=1/\fgair(n)$ and $W(\aopt)=(n-1)+1/\fgair(n)$.  Expanding out we have
	\begin{eqnarray}
		{\rm PoS}(\gee_{\fgair}^n) &\le & \frac{W(\ane)}{W(\aopt)} \\ 
		&=& \frac{1}{1+(n-1)\fgair(n)} \\
		&=& 1 - \frac{1}{{\frac{1}{(n-1)(n-1)!} + \sum_{i=0}^{n-1} \frac{1}{i!}}}.
	\end{eqnarray}
Accordingly, $	\lim_{n \rightarrow \infty} \ {\rm PoS}(\gee_{\fgair}^n) \le 1-1/e$. Nevertheless, by definition of price of anarchy and of price of stability we have $\lim_{n \rightarrow \infty} {\rm PoA}(\gee_{\fgair}^n)\le \lim_{n \rightarrow \infty}{\rm PoS}(\gee_{\fgair}^n)$, which coupled with $\lim_{n \rightarrow \infty} {\rm PoA}(\gee_{\fgair}^n)=1-1/e$ and  $\lim_{n \rightarrow \infty} {\rm Pos}(\gee_{\fgair}^n)\le 1-1/e$ completes the proof. $\hfill \Box$}\\

\DP{The section provides an initial analysis on the performance guarantees associated with three distinct distribution rules (see the right panel on Figure \ref{table:distribution-values} for a summary). First, note that there is never an incentive to utilize the equal share distribution rule as the marginal contribution distribution rule achieves the same price of anarchy guarantees while also ensuring a better price of stability ($1$ as opposed to $1/2$).  The comparison between the marginal contribution distribution rule and the Gairing distribution rule is not as straightforward.  If the goal is to have a better price of anarchy, then the Gairing distribution rule is the better choice and one inherits a suboptimal price of stability ($1-1/e$ as opposed to $1$). If the goal is to have a better price of stability, then the marginal contribution distribution rule is the better choice and one inherits a suboptimal price of anarchy ($1/2$ as opposed to $1-1/e$).  Whether or not there are alternative distribution yields that achieve a more desirable balance between these two efficiency guarantees is the focus of the ensuing sections.} 
\begin{figure}[h!] 
	\DP{\begin{center}	
		\begin{tabular}[b]{c|c|c|c}
			\!\!$j$\!&\!$\fes(j)$\!&\!$\fmc(j)$\!&\!$\fgair(j)$\!\!\\\hline
			1 & 1 & 1 & 1\\[-0.5mm]
			2 & 1/2 & 0 & 0.418\\[-0.5mm]
			3 & 1/3 & 0 & 0.254\\[-0.5mm]
			4 & 1/4 & 0 & 0.180\\[-0.5mm]
			5 & 1/5 & 0 & 0.139\\[-0.5mm]
			6 & 1/6 & 0 & 0.113\\[-0.5mm]
			7 & 1/7 & 0 & 0.095\\[-0.5mm]
			8 & 1/8 & 0 & 0.082\\[-0.5mm]
			9 & 1/9 & 0 & 0.072\\[-0.5mm]
			10 & 1/10 & 0 & 0.065
		\end{tabular}
		\input{figures/POAvsPOS_initial.tikz}	 \caption{\DP{On the left, distribution rule values for the equal share distribution rule in \eqref{eq:fes}, marginal cost distribution rule in \eqref{eq:fmc}, and the Gairing distribution rule in \eqref{eq:fgair} for the case when $n=10$. On the right, corresponding prices of anarchy and prices of stability as $n\rightarrow\infty$.}}
		\label{table:distribution-values}
	\end{center}}
	\vspace*{-8mm}
\end{figure}

\section{\DP{The Trade-off Between the Price of Stability and Price of Anarchy}}\label{sec:tradeoff}
\DP{The goal of this paper is to investigate the design of distribution rules that optimize the metrics introduced in \eqref{eq:poa2} and \eqref{eq:pos3}. The distribution rules highlighted in the previous section hint to the fact that there may be a fundamental trade-off between the price of anarchy and price of stability.  That is, optimizing the price of anarchy necessarily comes at the expense of the price of stability (and vice versa).  This section provides a characterization of the precise trade-off between these two important performance measures. 
\subsection{Characterizing the Trade-off} }

In this section we provide our first main result that characterizes the inherent tension between the price of anarchy and price of stability as design objectives in multiagent maximum coverage problems. \DP{For the remainder of the manuscript, we let ${\cal F}$ denote the set of functions $f:{\mathbb N} \rightarrow \arr$ where we set $f(1) = 1$, which is without loss of generality.  Further, let $\gee_f = \lim_{n \rightarrow \infty} \gee_f^n$ capture all maximum coverage games with an arbitrary number of players and a given distribution rule $f \in {\cal F}$. First, note that for any $f \in {\cal F}$ and any $n \geq 1$, the set of games $\gee_f^n$ is well defined, even though certain elements of $f$ might not utilized. Second, for any $f \in {\cal F}$ and $n \geq 1$, the price of anarchy (and similarly the price of stability) satisfies ${\rm PoA}(\gee_f) \leq {\rm PoA}(\gee^n_f)$.}

\begin{theorem}\label{thm:tradeoff}
	Consider the class of maximum coverage games introduced in Section~\ref{sec:modelmetrics}. The following holds:
	\begin{itemize}
		\item [(i)] The optimal price of anarchy satisfies 
		\begin{equation}\label{eq:t1}
		\max_{f\in {\cal F}} \ {\rm PoA}(\gee_f) = 1-1/e.
		\end{equation}
		\item[(ii)] Given a desired price of anarchy $\alpha \in [0,1/2]$, the best attainable price of stability satisfies  
		\begin{equation}
		\underset{f \in {\cal F} : {\rm PoA} (\gee_f)\ge\alpha}{\max} \ {\rm PoS}(\gee_f) =1.
		\end{equation}
		\item[(iii)] Given a desired price of anarchy $\alpha \in (1/2, 1-1/e]$ and $n\geq2$, the best attainable price of stability satisfies
		\begin{equation}\label{eq:t3}
		\underset{f \in {\cal F} : {\rm PoA} (\gee_f^n)\ge\alpha}{\max} \ {\rm PoS}(\gee_f^n) \leq  Z(\alpha,n),
		\end{equation}
		where $Z(\alpha,n)$ equals
		\begin{equation}\label{eq:t2}
		 \frac{1}{1+ 
			{\underset{1\le j \le n-1}{\max} \ j\,j!\left(1 - \left(\frac{1}{\alpha}-1\right) \sum_{i=1}^j \frac{1}{i!}\right)}}.
		\end{equation}
		The bound in \eqref{eq:t3} is satisfied with equality if we restrict attention to single-selection maximum coverage games. 

		\item[(iv)] Given a desired price of anarchy $\alpha = 1 -1/e$, the best attainable price of stability satisfies  
		\begin{equation}
		\underset{f \in {\cal F} : {\rm PoA} (\gee_f)\ge 1-1/e}{\max}\ {\rm PoS}(\gee_f) = 1-1/e.
		\label{eq:posinequality}
		\end{equation}
	\end{itemize} 
\end{theorem}

\DP{Theorem~\ref{thm:tradeoff} confirms the intuition gleaned in the previous section that there is indeed a trade-off between the price of anarchy (performance guarantees associated with the worst performing equilibrium) and price of stability (performance guarantees associated with the best performing equilibrium). Hence, there is an inherent tension between these two measures of efficiency as improving the performance of the worst equilibria necessarily comes at the expense of the performance of the best equilibria, and vice versa. The explicit trade-off given in  Theorem~\ref{thm:tradeoff} is illustrated in Figure \ref{fig:tradeoff}. In particular, Theorem~\ref{thm:tradeoff} establishes that there does not exist a distribution rule $f$ that attains joint price of stability and price of anarchy in the red region of the figure.  The expression of $Z(\alpha,n)$ defines this trade-off and satisfies $\lim_{n\rightarrow \infty} Z(1-1/e,n) = 1-1/e$. The last result demonstrates that this trade-off curve is tight in the context of single-selection maximum coverage games, i.e., there are distribution rules $f$ that achieve joint price of anarchy and price of stability guarantees highlighted on the curve $Z(\alpha,n)$. The specific distribution rules will be characterized in the proof of Theorem~\ref{thm:tradeoff}. }

\begin{figure}[t!] 
	\begin{center}
	\newlength\figureheight 
	\newlength\figurewidth
	\setlength\figureheight{6cm} 
	\setlength\figurewidth{6cm} 
	\input{figures/POAvsPOS.tikz}
	\caption{The figure provides an illustration of the inherent trade-off between the price of anarchy and price of stability.  First, note that the gray region is not achievable since $0 \le {\rm PoA}(\gee_f) \le {\rm PoS}(\gee_f)\le 1$ by definition. Theorem~\ref{thm:tradeoff} demonstrates that the red region is also \emph{not} achievable.  That is, there does not exist a distribution rule with joint price of anarchy and price of stability guarantees in the red region. For example, if the desired price of anarchy is $\alpha \leq 1/2$, then a price of stability of $1$ is attainable while meeting this price of anarchy demand.  However, if the desired price of anarchy is $\alpha = 1-1/e$, then a price of stability of $1$ is no longer attainable.  In fact, the best attainable price of stability is now also $1-1/e$.%
	}
		\label{fig:tradeoff}
	\end{center}
\end{figure}

\DP{Before delving into the proof of Theorem~\ref{thm:tradeoff}, we present an important result that characterizes the price of stability.  This results confirms the intuition of the structure of the worst-case games for the price of anarchy as shown in the proofs of Propositions~\ref{claim:eq-share} and \ref{claim:gair-share}.} 
\DP{
\begin{theorem}\label{t2}
	Let $n \geq 2$ and consider any distribution rule $f \in {\cal F}$. The price of stability associated with the induced family of games $\gee_f^n$ satisfies
	\begin{equation}\label{eq:pos2}
	{\rm PoS}(\gee_f^n) = \min_{1\le j \leq n } \left\{ \frac{1}{1+(j-1)  f(j) } \right\}.
	\end{equation}
	Further, the optimal price of stability satisfies 
\begin{equation}\label{eq:981}
	\max_{f \in {\cal F}} \ {\rm PoS}(\gee_f^n) = 1,
\end{equation}
	and a distribution rule that achieves this price of stability over the induced games $\gee_f^n$ is the marginal contribution distribution rule defined in \eqref{eq:fmc}.  Lastly, \eqref{eq:pos2} is satisfied with equality for single-selection covering games with $n$ agents.  
\end{theorem}}
\DP{Note that proving a bound on the price of stability is often technically challenging, as we need to perform a worst-case analysis over a reduced set of equilibria, i.e., over the best-performing equilibria. As a result of this, the proof of Theorem~\ref{t2} is non-trivial, and therefore deferred to the Appendix.}

\subsection{Proof of Theorem~\ref{thm:tradeoff}}
In this section we complete the Proof of Theorem~\ref{thm:tradeoff}. 
\DP{%
Part~(i) follows immediately from Lemma~\ref{t1}. Part~(ii) follows from Proposition~\ref{claim:mc-share}. We now turn the attention to Part~(iii). Towards this goal, let 
$\mathcal{F}(\alpha)$, $\alpha \in (0,1]$ be the family of distribution rules that guarantee a price of anarchy of at least $\alpha$, i.e., 
\[
\mathcal{F}(\alpha)=\{ f\in\mathcal{F} : {\rm{PoA}}(\gee_f^n) \geq \alpha \}.
\]
Additionally, let 
\begin{equation}\label{eq:posgivenalpha}
{\rm PoS}(\gee^n;\alpha) = \max_{f \in \mathcal{F}(\alpha)} \   {\rm PoS}(\gee_f^n) 
\end{equation}
be the best achievable price of stability given that the price of anarchy is guaranteed to exceed $\alpha$, where $\alpha \in (0,1)$.
}

	We now prove that \eqref{eq:t3} holds with the equality sign, when restricting to single-selection covering games. It follows immediately that \eqref{eq:t3} hold with the inequality sign for general (not necessarily single-selection) covering game.

	Restricting our attention to single-selection covering games, from Theorem~\ref{t2} and~\eqref{eq:posgivenalpha} we have
	\[
	{\rm PoS}(\gee^n; \alpha) = \max_{f \in {\cal F}(\alpha)} \min_{1\le j \leq n} \left\{ \frac{1}{1+ (j-1)  f(j)} \right\}
	\]
	or equivalently
	\begin{equation}\label{eq:l2}
	{\rm PoS}(\gee^n; \alpha) = \max_{f \in {\cal F}(\alpha)} \min_{1\le j \leq n-1} \left\{ \frac{1}{1+ j  f(j+1)} \right\},
	\end{equation}
	since for $j=1$ it is $(j-1)f(j)=0$, and $(j-1)f(j)\ge 0$ for all other $2\le j\le n$.
	Thanks to Lemma~\ref{t1}, the constraint ${\rm{PoA}}(\gee_f^n) \geq \alpha$ can be written as 
	\begin{alignat*}{4}
	& && j  f(j)  - \left(\frac{1}{\alpha} -1\right) \leq f(j+1) &&&\quad  1\le j \le n-1,\\
	& && jf(j+1) \leq \left(\frac{1}{\alpha}-1\right) &&&\quad  1\le j \le n-1,\\
	& && (j+1) f(j+1) - 1\leq \left(\frac{1}{\alpha}-1\right) &&&\quad  1\le j\le n-1.
	\end{alignat*} 
	Thus, the optimization problem in \eqref{eq:l2} is equivalent to
	\begin{alignat*}{4}
	\min_{f} & \underset{1\le j \le n-1}{\max} && j  f(j+1) &&& \\
	& \text{s.t.} && j  f(j)  - \left(\frac{1}{\alpha} -1\right) \leq f(j+1) &&&\quad  1\le j \le n-1,\\
	& && jf(j+1) \leq \frac{1}{\alpha}-1 &&&\quad  1\le j \le n-1,\\
	& && (j+1) f(j+1) \leq \frac{1}{\alpha} &&&\quad  1\le j\le n-1,\\
	& && f(j)\ge 0 &&&\quad  1\le j \le n,\hspace*{6mm}\\
	& && f(1) = 1. &&&
	\end{alignat*}
	Recall that $f(1)=1$. Recursively applying the first set of inequalities, it follows that for every $j$ with $1\le j\le n-1$
	\[
	jf(j+1)\ge j\,j! \biggl(1 - \left(\frac{1}{\alpha}-1\right) \sum_{i=1}^j \frac{1}{i!}\biggr).
	\]
	Since our objective is precisely to minimize the quantity $\max_{1\le j\le n-1} jf(j+1)$, we find a candidate solution by first solving for the distribution rule that satisfies the $n-1$ linear inequalities with equality, or set it to zero if the term $jf(j)-({1}/{\alpha}-1)$ is negative. Such a distribution rule can be computed recursively, and is of the form
	\[
	\hat{f}(j) = \max\biggl\{ (j-1)! \biggl(1 - \left(\frac{1}{\alpha}-1\right) \sum_{i=1}^{j-1} \frac{1}{i!}\biggr), 0 \biggr\}.
	\]
	By construction, the above distribution satisfies $\hat f(1)=1$ as well as the first and fourth set of constraints in the optimization problem above. In the following we verify that $\hat f$ also satisfies the remaining set of constraints, and thus is optimal. \DP{Note that the terms $\hat{f}(k)$, $k > n$, are irrelevant in the given analysis. We now show that $\hat{f}$ also satisfies the second and third set of constraints.}
	We begin with the second set of constraints, i.e., we wish to show $j\hat{f}(j+1) \leq \left({1}/{\alpha}-1\right)$ for all $\alpha$ with $\frac{1}{2}< \alpha\le \max _{f\in\mathcal{F}} {\rm PoA}(\gee_f^n)$ and for all $j$ with $1\le j \le n-1$. This is equivalent to
	\[
	j\,j!\biggl(1 - \left(\frac{1}{\alpha}-1\right) \sum_{i=1}^{j} \frac{1}{i!}\biggr)\le \frac{1}{\alpha}-1\,,
	\]
	as the result follows readily when $\hat f$ takes the value $\hat f (j)=0$. After some manipulation this is equivalent to showing
	\[
	\frac{1}{j\,j!}+\sum_{i=1}^j\frac{1}{i!}\ge\biggl(\frac{1}{\alpha}-1\biggr)^{-1}
	\]
	for all $\alpha$ with $\frac{1}{2}< \alpha\le \max _{f\in\mathcal{F}} {\rm PoA}(\gee_f^n)$ and for all $j$ with $1\le j \le n-1$.
	Since $\alpha\le \max _{f\in\mathcal{F}} {\rm PoA}(\gee_f^n)$, the result follows if we are able to prove the above inequality for $\alpha=\max _{f\in\mathcal{F}} {\rm PoA}(\gee_f^n)$, i.e.,  
	\begin{equation}
	\frac{1}{j\,j!}+\sum_{i=1}^j\frac{1}{i!}\ge\biggl(\frac{1}{\max _{f\in\mathcal{F}} {\rm PoA}(\gee_f^n)}-1\biggr)^{-1}.
	\label{eq:inequalityoptf}
	\end{equation}
	The expression for the optimal price of anarchy was found \cite[Corollary 1]{Gairing09}, and amounts to
	\[
	{\max _{f\in\mathcal{F}} {\rm PoA}(\gee_f^n)} = 1 - \frac{1}{\frac{1}{(n-1)(n-1)!} + \sum_{i=0}^{n-1} \frac{1}{i!}}.
	\]
	Replacing this in \eqref{eq:inequalityoptf}, we are left to prove that for all $1\le j\le n-1$
	\begin{equation}
	\label{eq:desiredresult}
	\frac{1}{j\,j!}+\sum_{i=0}^j\frac{1}{i!}\ge
	\frac{1}{(n-1)(n-1)!}+\sum_{i=0}^{n-1}\frac{1}{i!}\,.	
	\end{equation}
	To do so, let us define the function $g:\mathbb{N}\rightarrow\mathbb{R}$ as
	\[
	g(j)=\frac{1}{j\,j!}+\sum_{i=0}^j\frac{1}{i!},
	\]
	and observe that $g(j)$ is non-increasing in $j$. To see this note that 
	\begin{alignat*}{2}
	&g(j)\ge g(j+1)  &&\iff\\
	&\frac{1}{j\,j!}+\sum_{i=0}^j\frac{1}{i!}\ge
	\frac{1}{(j+1)\,(j+1)!}+\sum_{i=0}^{j+1}\frac{1}{i!}
	&&\iff\\
	&\frac{1}{j\,j!}-\frac{1}{(j+1)(j+1)!}\ge \frac{1}{(j+1)!} &&\iff\\
	&\frac{1}{j}\ge\frac{1}{j+1},
	\end{alignat*}
	which holds for any $j\in\mathbb{N}$. Thus, we conclude that for $j$ with $1\le j \le n-1$ it holds
	\[
	g(j)\ge g(n-1),
	\]
	which is precisely the desired result in \eqref{eq:desiredresult}, thus proving that the third set of inequalities is satisfied.
	
	We now turn our attention to the third set of constraints, and wish to prove that $(j+1)\hat{f}(j+1) \leq {1}/{\alpha}$ for all $\alpha$ with $\frac{1}{2}< \alpha\le \max _{f\in\mathcal{F}} {\rm PoA}(\gee_f^n)$ and for all $j$ with $1\le j \le n-1$. This follows directly as a consequence of $j\hat{f}(j+1)\le (1/\alpha-1)$ previously shown. Indeed 
	\[
	j\hat{f}(j+1)\le\frac{1}{\alpha}-1 \iff \hat f(j+1)\le \frac{1}{j}\left(\frac{1}{\alpha}-1\right), 
	\] 
	which implies the desired result
	\[
	\begin{split}
	(j+1)\hat f(j+1)&\le \frac{j+1}{j}\left(\frac{1}{\alpha}-1\right)=\left(\frac{1}{\alpha}-1\right)+\frac{1}{j}\left(\frac{1}{\alpha}-1\right)\\
	&\le \left(\frac{1}{\alpha}-1\right)+1 = \frac{1}{\alpha}\,,
	\end{split}
	\]
	where the last inequality follows from $j\ge 1$ and $1/\alpha-1\le 1$ (due to $\alpha > 1/2$). This concludes the above reasoning, and shows that $\hat f$ is a solution to \eqref{eq:l2}.
	Using Theorem~\ref{t2} on such distribution rule gives the desired result.

The result in Part (iv) can be proven by replacing $\alpha$ in \eqref{eq:t2} with the optimal price of anarchy from \cite{paccagnan2018arxiv2,Gairing09}, i.e., set
	\[
	\alpha^\star = 1 - \frac{1}{\frac{1}{(n-1)(n-1)!} + \sum_{i=0}^{n-1} \frac{1}{i!}}.
	\]
	After some manipulation, we obtain
	\[
	{\rm{PoS}}(\gee^n;\alpha^\star)	=  \frac{{\frac{1}{(n-1)(n-1)!} + \sum_{i=1}^{n-1} \frac{1}{i!}}}{{\frac{1}{(n-1)(n-1)!} + \sum_{i=0}^{n-1} \frac{1}{i!}}}=\max _{f\in\mathcal{F}} {\rm PoA}(\gee_f^n).
	\] Taking the limit as $n\to\infty$ gives the final result.
	$\hfill \Box$

\section{Using Information to Breach the Anarchy / Stability Frontier}
\label{sec:additionalinfo}

The previous section highlights a fundamental tension between the price of stability and price of anarchy for the given covering problem when restricted to local agent objective functions of the form \eqref{eq:ut}. In this section, we challenge the role of locality in these fundamental trade-offs. That is, we show how to move beyond the price of anarchy / price of stability frontier given in Theorem \ref{thm:tradeoff} if we allow the agents to condition their choice on a higher degree of system-level information. To do so, we restrict our attention to single-selection covering games.

With this goal in mind, we introduce a minimal and easily attainable piece of system-level information that can permit the realization of decision-making rules with efficiency guarantees beyond this frontier. To that end, for each allocation $a\in \mathcal{A}$ we define the information flow graph $(V,E)$ where each node of the graph represents an agent and we construct a directed edge $i\rightarrow j$ if $a_i \in \mathcal{A}_j$ for $i\ne j$ (no self loops). \DP{This describes the situation in which the resource selected by agent $i$ in allocation $a$ also belongs to the allocation set of agent $j$.} Based on this allocation-dependent graph, we define for each agent $i$ the set
$\nee_i(a)\subseteq N$ consisting of all the agents that
can reach $i$ through a path in the graph $(V,E)$.
Similarly, for each agent $i$ we define
\begin{equation}\label{eq:qidef}
Q_i(a)=(\cup _{j\in \nee_i(a)} \mathcal{A}_j)\cup \mathcal{A}_i,
\end{equation}
which consists of the union of $\mathcal{A}_i$ and all the sets of other agents that can reach $i$ through a path in the graph. An example is shown in Figure \ref{fig:enlargedsets}. Building upon this graph we define the following quantities:
\begin{eqnarray}
V_i(a) & = &  \max_{r \in \aee_i\setminus a_{-i}} v_{r}, \label{eq:videf}\\
x_i(a) & = & \max_{r \in Q_i(a)\setminus a} v_{r}.\label{eq:xidef}
\end{eqnarray}
The term $V_i(a)$ captures the highest valued resource in agent $i$'s choice set $\aee_i$ that is not covered by any agent. If the set $\mathcal{A}_i\setminus a_{-i}$ is empty, we set $V_i(a)=0$.
Similarly, the term $x_i(a)$ captures the highest-valued resource in the enlarged set $Q_i(a)$ not currently covered by any other agent. If the set $Q_i(a)\setminus a$ is empty, we set $x_i(a)=0$.

We are now ready to specify the information based covering game with a set of agents $N$, where each agent has an action set $\aee_i \subseteq \ree$.  Here, we consider a state-based distribution rule that toggles between the two extreme optimal distribution rules $\fgair$ and $\fmc$.  More formally, the distribution rule for agent $i$ is now of the form
\begin{equation}\label{eq:statebasedf}
f_i^{\rm sb}(a_i=r,a_{-i}) = \left\{ \begin{array}{ll} 
\fmc(|a|_r) &\text{if} \ V_i(a) \geq x_i(a), \\[.25cm]
\fgair(|a|_r) &\text{otherwise}, 
\end{array}\right. 
\end{equation}
and the corresponding utilities are given by $$
U_i(a_i=r,a_{-i})=v_r\cdot f_i^{\rm sb}(a),$$ as we allow the system-level information $x_i(a)$ and $V_i(a)$ to prescribe which distribution rule each agent applies. We denote with $\fsb=\{f_i^{\rm sb}\}_{i\in N}$ the collection of distribution rules in \eqref{eq:statebasedf} and informally refer to it as the state-based distribution rule. Throughout, we express the distribution rule as merely $f_i^{\rm sb}(a)$ instead of $f_i^{\rm sb}(x_i(a),V_i(a))$ for brevity. 

The next theorem demonstrates how $\fsb$ attains performance guarantees beyond the price of stability / price of anarchy frontier established in Theorem~\ref{thm:tradeoff}.

\begin{theorem}\label{t4}
	Consider any single-selection maximum coverage game with a state-based distribution rule $\fsb$ as defined above.  First, an equilibrium is guaranteed to exist in any game $G \in \gee_{\fsb}$.  Furthermore, the price of anarchy and price of stability associated with the induced family of games $\gee_{\fsb}$ are
	\[
	\begin{split}	
	{\rm PoS}(\gee_{\fsb}) =& 1,\\
	{\rm PoA}(\gee_{\fsb}) =& \max_{f\in\mathcal{F}} \ {\rm PoA}(\gee_{f})= 1 - \frac{1}{e}.
	\end{split}
	\]
\end{theorem}

\begin{figure}[t]
	\centering
	\includegraphics[width=0.3\columnwidth]{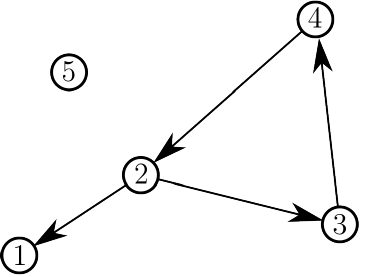}
	\\[3mm]
	\begin{tabular}[b]{ccc}\hline
		Agent $i$ & $\nee_i(a)$ & $Q_i(a)$ \\ \hline
		1 & \{2,\,3,\,4\} & $\mathcal{A}_1\cup\mathcal{A}_2\cup \mathcal{A}_3\cup \mathcal{A}_4$  \\
		2 & \{3,\,4\} & $\mathcal{A}_2\cup \mathcal{A}_3\cup \mathcal{A}_4$ \\
		3 & \{2,\,4\} &  $\mathcal{A}_2\cup \mathcal{A}_3\cup \mathcal{A}_4$ \\
		4 & \{2,\,3\} &  $\mathcal{A}_2\cup \mathcal{A}_3\cup \mathcal{A}_4$ \\
		5 & $\emptyset$ &  $\mathcal{A}_5$ \\
		\hline
	\end{tabular}
	\caption{{Example of graph $(V,E)$.} The sets $\nee_i(a)$ and $Q_i(a)$ are enumerated for each agent in the table on the right. Observe that the graph $(V,E)$ as well as the sets $\nee_i(a)$ and $Q_i(a)$ are allocation dependent.}
	\label{fig:enlargedsets}
\end{figure}

Recall from Theorem~\ref{thm:tradeoff} that a consequence of attaining a price of anarchy of $1-1/e$ was a price of stability also of $1-1/e$ and this was achieved by $\fgair$ defined in \eqref{eq:fgair}. Using the state-based rule given in \eqref{eq:statebasedf}, a system designer can achieve the optimal price of anarchy without any consequences for the price of stability. Hence, the identified piece of system-level information was crucial for moving beyond the inherited performance limitations by adhering to our notion of local information. Whether alternative forms of system-level information could move us beyond these guarantees, or achieve these guarantees with less information, is an open question.

\subsection{Proof of Theorem~\ref{t4}}

We will now present the proof of Theorem~\ref{t4}.  For readability, we will split the proof into the following two subsections focusing on the price of anarchy and price of stability respectively.

\subsubsection{Proof of ${\rm PoS}$ Result}

We begin our proof with a lemma that identifies an important structure regarding the state based distribution rule $\fsb$. \\ 
\begin{lemma}\label{lemma:always}
	Let $a$ be any allocation. Then for each agent $i \in N$, one of the following two statements is true.  
	\begin{itemize}
		\item[-] $U_i(a_i'=r, a_{-i}) = v_r \cdot \fmc(|a_{-i}|+1), \ \forall a_i' \in \aee_i$; or 
		\item[-] $U_i(a_i'=r, a_{-i}) = v_r \cdot \fgair(|a_{-i}|+1), \ \forall a_i' \in \aee_i$.
	\end{itemize}
\end{lemma}

Informally, this lemma states that for a given allocation $a$, the state based distribution rule will either evaluate every resource at $\fgair$ or $\fmc$ for a given agent $i$.  

\pad 

{\it Proof of Lemma~\ref{lemma:always}}
	Let $a$ be any allocation and $i$ be any agent. Extend the definition of $x_i(a)$ in \eqref{eq:xidef} as
	\begin{eqnarray}
	y_i(a) &=& \max_{r \in (Q_i(a) \setminus a) \cap \aee_i} v_r, \\
	z_i(a) &=& \max_{r \in (Q_i(a) \setminus a) \setminus \aee_i} v_r,	
	\end{eqnarray} 
	and note that $x_i(a) = \max \{y_i(a), z_i(a)\}$. First observe that for any $a_i' \in \aee_i$: (i) $y_i(a) \leq V_i(a)$; (ii) $z_i(a) = z_i(a_i', a_{-i})$; and (iii) $V_i(a) = V_i(a_i', a_{-i})$. Accordingly, $x_i(a) > V_i(a)$ if and only if $x_i(a) = z_i(a)$.  Consequently, $x_i(a) > V_i(a)$ if and only if $x_i(a_i', a_{-i}) > V_i(a_i', a_{-i})$ which completes the proof.  
$\hfill \Box$

We will now prove that an equilibrium exists and the price of stability is $1$.  In particular, we will show that the optimal allocation $a^{\rm opt}$ is in fact an equilibrium.  

\begin{lemma}\label{lemma:optimalisne}
	Consider any single-selection maximum coverage game with a state-based distribution rule $\fsb$.  The optimal allocation $\aopt$ is an equilibrium. 
\end{lemma}

{\it Proof of Lemma~\ref{lemma:optimalisne}}
	Let $a^{\rm opt}$ be an optimal allocation.  We begin by showing that $V_i(a^{\rm opt}) \geq x_i(a^{\rm opt})$ for all $i \in N$.  Suppose this was not the case, and there exists an agent $i$ such that $V_i(a^{\rm opt}) < x_i(a^{\rm opt})$. This implies that there exists a resource $r \in Q_i(\aopt)\setminus\aopt$ such that $v_r > v_{\tilde r}=\max_{r_0 \in \aee_i : |a_{-i}^{\rm opt}| = 0} v_{r_0}$. By definition of $Q_i(a)$ and ${\cal N}_i(a)$ there exists a sequence of players $\{i_0, i_1, \dots, i_m\}$ such that $a_{i_k}^{\rm opt} \in \aee_{i_{k-1}}$ for all $k \geq 1$, and $r \in \aee_{i_m}$. Hence, consider a new allocation where $a_{i_k} = a_{i_{k+1}}^{\rm opt}$ for all $k \in \{0, \dots, m-1\}$ and $a_{i_m}=r$, where $a_j = a_j^{\rm opt}$ for all other agents $j\notin\{i_0, i_1, \dots, i_m\}$.  The welfare of this allocation is $W(a) \geq W(a^{\rm opt})+v_r - v_{\tilde{r}} > W(a^{\rm opt})$, which contradicts the optimality of $a^{\rm opt}$. This means that every agent will be using the marginal cost distribution rule to evaluate its utility in the allocation $\aopt$.  Now, suppose $a^{\rm opt}$ is not an equilibrium for sake of contradiction. This means, that there exists an agent $i$ with an action $a_i \in \aee$ such that $U_i(a_i, a_{-i}^{\rm opt}) > U_i(\aopt)$. Since agents are using the marginal contribution distribution rule, which follows from $V_i(a^{\rm opt}) \geq x_i(a^{\rm opt})$, we have 
	$$ W(a_i, a_{-i}^{\rm opt}) = W(a^{\rm opt}) - U_i(a^{\rm opt}) + U_i(a_i, a_{-i}^{\rm opt}) > W(a^{\rm opt}), $$
	which contradicts the optimality of $a^{\rm opt}$. This completes the proof.  
$\hfill \Box$

\subsubsection{Informal Discussion of ${\rm PoA}$ Result}
In the following we give an informal discussion for the price of anarchy result. This will be followed by a formal proof.

Consider a game $G = ( N, \ree, \{\aee_i\}, \{f_i\},\{v_r\})$. Let $a^{\rm ne}$ be any equilibrium of the game $G$.  A crucial part of the forthcoming analysis will center on a new game $G'$ derived from the original game $G$ and the equilibrium $a^{\rm ne}$, i.e., 
$$ (G, a^{\rm ne}) \ \longrightarrow G'. $$
This new game $G'$ possesses the identical player set, resource set, and valuations of the resources as the game $G$. The difference between the games are (i) the action sets and (ii) the new game $G'$ employs the Gairing distribution rule, $\fgair$, as opposed to the state-based distribution rule $\fsb$. Informally, the proof proceeds in the following two steps:

\pad

\noindent \textbf{{-- Step 1}:} 
We prove that the equilibrium $a^{\rm ne}$ of $G$ is also an equilibrium of $G'$.  Since the player set, resource set, and valuations of the resources are unchanged we have that 
$$W(a^{\rm ne};G')=W(a^{\rm ne};G),$$
where we write the notation $W(a^{\rm ne};G')$ to mean the welfare accrued at the allocation $a^{\rm ne}$ in the game $G'$.

\pad

\noindent \textbf{{-- Step 2}:} 
We show that the optimal allocation $\aoptprime$ in the new game $G'$ is at least as good as the optimal allocation in the original game $G'$, i.e., 
$$W(\aoptprime;G') \geq W(a^{\rm opt};G).$$
Combining the results from Step~1 and Step~2 give us
$$ \frac{W(a^{\rm ne};G)}{W(a^{\rm opt};G)} \geq \frac{W(a^{\rm ne};G')}{W(\aoptprime;G')} \geq {\rm PoA}(\mathcal{G}_{\fgair}), $$
where the last inequality follows from Lemma~\ref{t1} since $G'$ employs $\fgair$.

\subsubsection{Proof of ${\rm PoA}$ Result}

In this section we present the formal proof for the price of anarchy result pertaining to the state based design.  We begin by providing the formal details on the game construction highlighted above.  

\subsubsection*{Construction of game $G'$} 

We will now provide the construction of the new game $G'$ from the game $G$ and the equilibrium $a^{\rm ne}$. We begin with some notation that we will use to construct the new agents' action sets in the game $G'$. For each $i \in N$, let $r_i = a_i^{\rm ne}$ and define
\begin{equation}\label{eq:H}
H_i= \{r \in \aee_i: v_{r_i}  \fgair(|a^{\rm ne}|_{r_i}) < v_{r}  \fgair(1+|a_{-i}^{\rm ne}|_r)\},
\end{equation}
to be the set of resources that will give a strictly better payoff to agent $i$ if the agent was to use $\fgair$ instead of $\fsb$.  Further, denote 
\begin{equation}
\mathcal{I} = \{ i \in N: H_i \neq \emptyset \},
\end{equation} 
as the set of agents that would move from the equilibrium $\ane$ if they were to use $\fgair$ instead of $\fsb$. Finally, for each agent~$i \in \mathcal{I}$, let
\begin{equation}\label{eq:B}
B_i = \{r' \in Q_{i}(\ane) : |\ane|_{r'} = 0 \}\,,
\end{equation}
i.e. the enlarged set of resources of \eqref{eq:qidef} that are not chosen by anyone at equilibrium. 

We are now ready to construct the new game $G'$.  As noted above, the agent set, resource set, and resource valuations are identical to those in $G$. The new action set of each agent $i$ in game $G'$ is defined as 
\[
\aee_i'=
\begin{cases}	
&(\aee_i \backslash H_i )\cup B_i \cup \emptyset  \qquad \text{if}\quad  i\in\mathcal{I}\,,\\
&~\aee_i  \hspace*{33mm} \text{otherwise}\,. 
\end{cases}
\]
Lastly, the utility functions of the agents are derived using $\fgair$.  We denote this game $G'$ by the tuple $G' = ( N, \ree, \{\aee_i'\}, \fgair,\{v_r\})$.

\vspace{.3cm}

{\it Formal Proof of Step 1:} 
In the first part of the proof, we will establish that the equilibrium allocation $a^{\rm ne}$ of game $G$ is also an equilibrium allocation of game $G'$.  By definition of the action sets $\{\aee_i'\}$, we know that $a_i^{\rm ne} \in \aee_i'$ for all agents $i \in N$.  It remains to show that for any agent $i \in N$
\begin{equation}\label{eq:951}
U_i(a_i^{\rm ne}, a_{-i}^{\rm ne}; G') \geq U_i(a_i', a_{-i}^{\rm ne}; G'), \ \forall a_i' \in \aee_i',
\end{equation} 
where we use the notation $U_i(a; G')$ to denote the utility of agent $i$ for allocation $a$ in the game $G'$. 

Based on the above definition of the action sets $\{\aee_i'\}$, we only need to concentrate our attention on agents $i \in \mathcal{I}$ with choices $a_i' \in B_i$, as \eqref{eq:951} follows immediately in the other cases.  Since $a^{\rm ne}$ is an equilibrium of game $G$, then by Lemma~\ref{lemma:always} we know that each agent $i \in \mathcal{I}$ employs the marginal contribution distribution rule everywhere.  Accordingly, we have
\begin{equation}\label{eq:952}
U_i(a_i^{\rm ne}, a_{-i}^{\rm ne}; G) = \max_{r \in \aee_i\setminus a_{-i}^{\rm ne}} v_r =  V_i(a^{\rm ne}) \geq x_i(a^{\rm ne}),
\end{equation} 
where the first equality follows from the equilibrium conditions coupled with the use of the marginal contribution distribution rule and the inequality follows from the use of the marginal contribution distribution rule and \eqref{eq:statebasedf}.

We will conclude the proof by a case study on the potential values of $U_i(a_i^{\rm ne}, a_{-i}^{\rm ne}; G)$.  For brevity in the forthcoming arguments we let $r = a_i^{\rm ne}$.

\pad

\noindent -- {Case (i):} Suppose $U_i(a_i^{\rm ne}, a_{-i}^{\rm ne}; G)=0$.  In this case, we have that $x_i(a^{\rm ne}) = 0$ from \eqref{eq:952} which implies that any resource $r' \in B_i$ has value $v_{r'}=0$. Hence, $U_i(a_i^{\rm ne}, a_{-i}^{\rm ne}; G') = U_i(r', a_{-i}^{\rm ne}; G') = 0$ and we are done.  

\pad

\noindent -- {Case (ii):} Suppose $U_i(a_i^{\rm ne}, a_{-i}^{\rm ne}; G)>0$.  Based on the definition of the marginal cost distribution rule, this implies that $|a^{\rm ne}|_r = 1$, and hence
$$ U_i(a^{\rm ne};G) = v_r  \fmc(|a^{\rm ne}|_r) = v_r  \fgair(|a^{\rm ne}|_r) = U_i(a^{\rm ne};G'),$$  
which follows from the fact that $\fmc(1) = \fgair(1)=1$. Hence, $U_i(a^{\rm ne};G')=v_r$. For any $r' \in B_i$, we have 
$$ U_i(r', a_i^{\rm ne};G') = v_{r'} \fgair(1+|a^{\rm ne}_{-i}|_{r'}) = v_{r'} \leq x_i(a^{\rm ne}),$$
where the last inequality follows for the definition of $x_i(a^{\rm ne})$.  Combining with \eqref{eq:952} gives us 
$$U_i(a_i^{\rm ne}, a_{-i}^{\rm ne}; G') \geq U_i(r', a_{-i}^{\rm ne}; G'),$$
which completes the proof.  $\hfill\Box$

\vspace{.3cm}

{\it Formal Proof of Step 2:}	
We begin with a lemma that highlights a structure associated with the action sets $\{\aee_i'\}$ in the new constructed game $G'$.

\begin{lemma}\label{lem:v}
	If $r \in \aee_i' \setminus \aee_i$ for some agent $i \in N$ and resource $r \in \mathcal{R}$, then there exists an agent $j \neq i$ such that $a_j^{\rm ne} = r$ and consequently $r \in \aee_j'$.  
\end{lemma}

\pad

{\it Proof of Lemma~\ref{lem:v}}
	Suppose $r \in \aee_i' \setminus \aee_i$ for some agent $i \in N$ and resource $r \in \mathcal{R}$.   Then, $r \in H_i$ by definition of the set $\aee_i'$. By Lemma~\ref{lemma:always}, each agent must either be a marginal contribution agent, i.e., uses $\fmc$ at all resources, or a Gairing agent, i.e., uses $\fgair$ at all resources.  Since, $i \in {\cal I}$ and $a^{\rm ne}$ is an equilibrium, agent $i$ must be a marginal contribution agent, i.e., 
	$$ U_i(a^{\rm ne}_{i} = \tilde{r}, a^{\rm ne}_{-i}) = v_{\tilde{r}} \cdot \fmc(|a^{\rm ne}|_{\tilde{r}}) \geq v_{r} \cdot \fmc(|a^{\rm ne}_{-i}|_{r}+1), $$
	where the inequality follows from the equilibrium conditions.  	Suppose by contradiction that $|a_{-i}^{\rm ne}|_{r}=0$.  In this case we have 
	$$ v_{r} \cdot \fmc(|a^{\rm ne}_{-i}|_{r}+1) = v_{r} \cdot \fgair(|a^{\rm ne}_{-i}|_{r}+1). $$
	which follows from the fact that $f^{\rm mc}(1)=\fgair(1)=1$. Since $\fgair(k) \geq \fmc(k)$ for all $k\geq 1$, this implies 
	$$ v_{\tilde{r}} \cdot \fgair(|a^{\rm ne}|_{\tilde{r}}) \geq v_{\tilde{r}} \cdot \fmc(|a^{\rm ne}|_{\tilde{r}}) \ge v_{r} \cdot \fgair(|a^{\rm ne}_{-i}|_{r}+1). $$
	Hence, $r \notin H_i$ leading to the contradiction.  This completes the proof.  
$\hfill\Box$

\pad
We exploit the result of Lemma~\ref{lem:v} to prove Step 2.
	Towards this goal, we construct an allocation $a \in \aee'$ that satisfies $W(a;G') = W(a^{\rm opt};G)$ where $\aopt \in\arg\max_{a\in\mathcal{A}}W(a;G)$.  We begin with an initial allocation $a$ where for each agent $i \in N$ 
	\[
	a=
	\begin{cases}
	\aopt_i \quad &\text{if } \aopt_i\in \mathcal{A}'_i\,,\\
	\emptyset \quad &\text{else}\,,
	\end{cases}
	\]
	That is, we assign each agent the agent's optimal allocation choice if it is available to them in the new action set $\aee'_i$. If all agents received their optimal choice, then the proof is complete.

	If this is not the case, then there will be a set of uncovered resources $\mathcal{U}(a)=\{r \in a^{\rm opt} : |a|_r = 0\}$ which we denote by $\mathcal{U}(a) = \{r_1, \dots, r_m\}$. We will now argue that we can construct a new allocation $a'$ that covers one additional resource from the set $\mathcal{U}(a)$, i.e., $\card{\mathcal{U}(a')} = \card{\mathcal{U}(a)} - 1$ and $a \subseteq a'$, where we denote with $\card{\mathcal{U}(a)}$ the cardinality of $\mathcal{U}(a)$.
	
	To that end, consider any uncovered resource $r_0 \in \mathcal{U}(a)$. By definition, there exists an agent $i_0 \in N$ such that $a_{i_0}^{\rm opt} = r_0$ but $r_0 \notin \aee_{i_0}'$. Consequently, we have that $r_0 \in H_{i_0}$ and by Lemma~\ref{lem:v} we know that there exists an agent $i_1 \neq i_0$ such that $a_{i_1}^{\rm ne}=r_0$.  Since $a_{i_1}^{\rm ne}=r_0$ we also have that $r_0 \in \aee_{i_1}'$ by definition.  We now analyze the following three cases:

	\vspace{.2cm}
	
	\noindent -- Case 1: Suppose $a_{i_1} = \emptyset$. Then define a new allocation $a'_{i_1} = r_0$ and $a'_{j} = a_{j}$ for all $j \neq i_1$, and we are done.

	\vspace{.2cm}
	
	\noindent -- Case 2: Suppose $a_{i_1} = r_1$ and $|a^{\rm ne}|_{r_1} = 0$, meaning that there are no agents at the resource $r_1$ in the equilibrium allocation. Then by definition $r_1 \in B_{i_0}$ and $r_1\in \aee_{i_0}'$.  Define the allocation $a'_{i_0} = r_1$, $a'_{i_1} = r_0$, and $a_j' = a_j$ for all $j \neq i_0, i_1$ and we are done.

	\vspace{.2cm}
	
	\noindent -- Case 3: Suppose $a_{i_1} = r_1$ and $|a^{\rm ne}|_{r_1} > 0$, meaning that there are agents at the resource $r_1$ in the equilibrium allocation. Select any agent $i_2$ such that $a_{i_2}^{\rm ne} = r_1$. 
	\begin{enumerate}
		\item[(i)] If $a_{i_2} = \emptyset$, then consider the allocation $a'_{i_1} = r_0$, $a'_{i_2} = r_1$ and $a_j' = a_j$ for all $j \neq i_1, i_2$ and we are done.
		\item[(ii)] Otherwise, if $a_{i_2} = r_2$, then let $a'_{i_1} = r_0$, $a'_{i_2} = r_1$, and repeat Case 2 or Case 3 depending on whether $|a^{\rm ne}|_{r_2} = 0$ or $|a^{\rm ne}|_{r_2} > 0$.  Note that Case~3-(ii) can be repeated at most $n$ iterations until an alternative case that terminates is reached.  To see this, note that each time an agent is given a new choice in this process, i.e., $a_{i} \rightarrow a_i' \neq a_{i}$, the agent's new choice is the agent's equilibrium choice, i.e., $a_i' = a^{\rm ne}_i$.  Therefore, once an agent is assigned a new choice, the agent will never be reassigned in this process. 
	\end{enumerate}
	
	Starting from $a$ as defined above, the above process results in a new allocation $a'$ that satisfies $\card{\mathcal{U}(a')} = \card{\mathcal{U}(a)} - 1$ and $a \subseteq a'$.  As with the allocation $a$, the allocation $a'$ satisfies $\max_r |a'|_r = 1$ and $a' \subseteq a^{\rm opt}$. If $a' = a^{\rm opt}$, we are done.  Otherwise, we can repeat the process depicted above to generate a new allocation $a''$ such that $\card{\mathcal{U}(a'')} = \card{\mathcal{U}(a')} - 1$ as nowhere in the process did we rely on the fact that $a_i = a_i^{\rm opt}$.  Repeating these arguments recursively provides the result. $\hfill\Box$

\section{Conclusions}

How should a system operator design a networked architecture?  The answer to this question is non-trivial and involves weighing the advantages and disadvantages associated with different design choices.  In this paper we highlight one novel trade-off pertaining to the worst-case and best-case performance guarantees in distributed maximum coverage problems with local information. Further, we demonstrate how a system designer can move beyond these trade-offs by equipping the agents with additional information about the system.   
Fully realizing the potential of multiagent systems requires the pursuit of a more formal understanding of the inherent limitations and performance trade-offs associated with networked architectures. While this paper focused purely on two performance measures, other metrics of interest include convergence rates, robustness to adversaries, and fairness. In each of these settings, it is imperative that a system operator fully understands the role of information within these trade-offs.  Only then, will a system operator be able to effectively balance the potential performance gains with the communication costs associated with propagating additional information through the system.

\section{Appendix}
\subsection{\DP{Proof of Theorem~\ref{t2}}}
\DP{This section is dedicated to the proof of Theorem~\ref{t2}.  We will prove the result on the price of stability through a series of intermediate lemmas. We begin by observing that any game $G$ in the class $\gee_f^n$ is a congestion game, and thus is a potential game as introduced in  \cite{Monderer96}, with a potential function $\phi:\aee \rightarrow \arr$ of the form\footnote{A proof of this can be found in \cite{Monderer96}, where the potential function $\phi$ is also defined.} 
\begin{equation}\label{eq:potential}
\phi(a) = \sum_{r\in\mathcal{R}}\sum_{j=1}^{|a|_r} v_r  f(j).
\end{equation}
Further recall that an equilibrium is guaranteed to exist in any potential game \cite{Monderer96}, and one such equilibrium is the allocation that maximizes the potential function $\phi$, i.e., $\ane \in \arg \max_{a \in \aee} \phi(a)$.}

To prove \eqref{eq:pos2}, we restrict our attention to the set of single-selection covering games, where the optimal allocation is disjoint, i.e., $a_i^{\rm opt} \neq a_j^{\rm opt}$, for any $i \neq j$.  We denote such games by the set $\bar{\gee}_f^n \subset \gee_f^n$. Our first lemma, stated without proof for brevity, demonstrates that restricting attention to single-selection covering games where the optimal allocation is disjoint is sufficient for characterizing the price of stability in single-selection covering games.  

\begin{lemma}
	\DP{Let $\bar{\gee}_f^n$ denote the set of $n$-agent single-selection maximum coverage games with distribution rule $f \in {\cal F}$}.  %
	It holds ${\rm PoS}(\bar{\gee}_f^n )={\rm PoS}({\gee}_f^n)$.
\end{lemma}

We will now proceed with a series of claims to demonstrate that ${\rm PoS}(\bar{\gee}_f^n)$ satisfies  \eqref{eq:pos2} with equality. %
The central part of the proof involves focusing on the equilibrium which maximizes the potential function in \eqref{eq:potential} in the considered class of games $\bar{\gee}^n_f$.  From this specific equilibrium, we consider a sequence of allocations taking the form $a^0 = \ane$ and $a^k = (\aopt_{i(k)}, a_{-i(k)}^{k-1})$ for all $k \in \{1, \dots, m\}$ where $i(k)$ is the deviating player in the $k$-th profile.  \DP{That is, any two adjacent allocations $a^{k-1}$ and $a^k$ differ by at most one agent's choice, namely agent $i(k)$ that switches from $a^{\rm ne}_{i(k)}$ to $a^{\rm opt}_{i(k)}$. In the case that $a^{\rm ne}_{i(k)} = a^{\rm opt}_{i(k)}$, the allocations are the same, i.e.,  $a^k = a^{k+1}$.}  The selection of the deviating players ${\cal I} = \{i(1), \dots, i(m)\}$ is chosen according to the following rules:  
\begin{enumerate}
	\item[(i)]  Let $i(1) \in N$ be any arbitrary player. 
	\item[(ii)] For each $k \geq 1$, if $\aopt_{i(k)} =\ane_{i(1)}$ or $\aopt_{i(k)} \notin \ane$ then the sequence is terminated.
	\item [(iii)] Otherwise, let $i(k+1)$ be any agent in the set $\{j\in N\,:\,a^k_j=\aopt_{i(k)}\}$ and repeat.
\end{enumerate}

\DP{The sequence of allocations defined above is merely employed as a convenient mathematical formulation to derive a relationship between $W(\ane)$ and $W(\aopt)$.  The first part of this derivation is provided in the following lemma.  }

\begin{lemma}\label{cl:1}
	Define $Q = \cup_{i \in {\cal I}} \ \ane_i$ and $\bar{Q} = \cup_{i \in {\cal I}} \ \aopt_i$. Then 
	\begin{equation} \label{eq:sum}
	\sum_{i \in {\cal I}} U_i(\ane) \geq \sum_{r \in Q \cap \bar{Q}} v_r f(|\ane|_r) + \sum_{r \in \bar{Q} \setminus Q} v_r.
	\end{equation}
\end{lemma}

{\it Proof of Lemma \ref{cl:1}}
	We begin with two observations on the above sequence of allocations: (a) the sequence of allocations can continue at most $n$ steps due to the disjointness of $\aopt$  and (b) if the sequence continues, it must be that %
	$\ane_{i(k+1)}=\aopt_{i(k)}$.  Observation (b) ensures us that 
	\begin{equation}\label{eq:sumtozero}
	\psi = \sum_{k=1}^{m-1} U_{i(k+1)}(a^k) - U_{i(k)}(a^k)  = 0. 
	\end{equation}
	Accordingly, we have that 
	\begin{eqnarray*}
		\phi(a^0) - \phi(a^m) &=&
		\sum_{k=0}^{m-1}\phi(a^k)-\phi(a^{k+1})\\
		&=& \sum_{k=0}^{m-1} U_{i(k+1)}(a^k) - U_{i(k+1)}(a^{k+1})   \\
		&=&U_{i(1)}(a^0)-U_{i(m)}(a^m) + \psi \\
		&=& U_{i(1)}(a^0)-U_{i(m)}(a^m)\ge0\,. 
	\end{eqnarray*}
	The first and third equalities follow by rearranging the terms. The second equality can be shown using the definition of $\phi$ as in~\eqref{eq:potential}; the last equality follows by \eqref{eq:sumtozero}.
	The inequality derives from the fact that $a^0=\ane$ optimizes the potential function.  Thanks to observation (b), one can show that 
	\[
	Q\setminus\bar Q = \ane_{i(1)}\setminus\aopt_{i(m)}
	\quad\text{and}\quad
	\bar Q\setminus Q = \aopt_{i(m)}\setminus \ane_{i(1)}\,.
	\]
	If $Q \setminus \bar{Q} \neq \emptyset$, it must be that $\ane_{i(1)}\neq \aopt_{i(m)}$ so that $Q\setminus\bar Q=\ane_{i(1)}$ and $\bar Q\setminus Q=\aopt_{i(m)}$\,.
	Using $\ane_{i(1)}\neq \aopt_{i(m)}$ in condition (ii), tells us that $\aopt_{i(m)}\notin \ane$ and thus the resource $\aopt_{i(m)}$ is not chosen by anyone else in the allocation $a^m$. Thus, when $Q \setminus \bar{Q} \neq \emptyset$,
	\begin{equation}
	U_{i(1)}(a^0) - U_{i(m)}(a^{m}) = \sum_{r \in Q \setminus \bar{Q}} v_r f(|\ane|_r) - \sum_{r \in \bar{Q} \setminus Q} v_r\ge 0.
	\end{equation}
	When $Q \setminus \bar{Q} = \emptyset$, also $\bar Q \setminus {Q} = \emptyset$ and thus the previous inequality still holds.
	Rearranging the terms and adding $\sum_{r \in Q \cap \bar{Q}} v_r f(|\ane|_r)$ to each side gives us
	\begin{equation}\label{eq:intermediate}
	\sum_{r \in Q} v_r f(|\ane|_r) \geq \sum_{r \in Q \cap \bar{Q}} v_r f(|\ane|_r) + \sum_{r \in \bar{Q} \setminus Q} v_r\,.
	\end{equation}
	Finally note that 
	\begin{eqnarray*}
		\sum_{i \in {\cal I}} U_i(\ane) \ge \sum_{r \in Q} v_r f(|\ane|_r)
	\end{eqnarray*}
	which together with \eqref{eq:intermediate} completes the proof.  
	$\hfill \Box$

Our next lemma shows that there exists a collection of disjoint sequences that covers all players in $N$. We will express a sequence merely by the deviating player set ${\cal I}$ with the understanding that this set uniquely determines the sequence of allocations.

\begin{lemma}\label{cl:2}
	There exists a collection of deviating players ${\cal I}^1, \dots, {\cal I}^p$ chosen according to the process described above such that $\cup_{k}\,{\cal I}^k = N$ and ${\cal I}^j \cap {\cal I}^k = \emptyset$ for any $j \neq k$.    
\end{lemma}

{\it Proof of Lemma \ref{cl:2}}
	Suppose ${\cal I}^1, {\cal I}^2, \dots, {\cal I}^{k}$ represent the first $k$ sequences of deviating players. Further assume that they are all disjoint. Choose some player $i \in N \setminus \cup_k {\cal I}^k$ to start the $(k+1)$-th sequence.  If no such player exists, we are done. Otherwise, construct the sequence according to the process depicted above. If the sequence terminates without selecting a player in $\cup_k {\cal I}^k$, then repeat this process to generate the $(k+2)$-th sequence.  Otherwise, let $i^{k+1}(j)$, $j\geq 2$, denote the first player in the $(k+1)$-th sequence contained in the set $\cup_k {\cal I}^k$.  Since $\aopt_i\neq\aopt_j$ (for $i\neq j$), this player must be %
	 the first player in a previous sequence.  Suppose this player is $i^\ell(1)$, where $\ell \in \{1, \dots, k\}$.  If this is the case, replace the $\ell$-th sequence with $\{i^{k+1}(1), \dots, i^{k+1}(j-1), {\cal I}^\ell\}$ which is a valid sequence and disjoint from the others. Then repeat the process above to choose the $(k+1)$-th sequence.  Note that this process can continue at most $n$-steps and will always result in a collection of disjoint sequences that cover all players in $N$.  This completes the proof. 
	$\hfill \Box$

In the following we complete the proof of Theorem~\ref{t2}, by means of Lemmas \ref{cl:1} and \ref{cl:2}.

{\it Proof of Theorem~\ref{t2}:}
	We being showing a lower bound on the price of stability.  Let ${\cal I}^1, \dots, {\cal I}^p$ denote a collection of deviating players that satisfies Lemma~\ref{cl:2}.  Further, let $Q^k$ and $\bar{Q}^k$ be defined as above for each sequence $k=1, \dots, p$.   Using the result \eqref{eq:sum} from Lemma~\ref{cl:1}, we have
	\begin{eqnarray*} 
		\sum_{ i \in N} U_i(\ane) 
		& = & \sum_{k=1}^p \sum_{i \in {\cal I}^k} U_i(\ane) \\ 
		& \geq & \sum_{k=1}^p \left(\sum_{r \in Q^k \cap \bar{Q}^k}  v_r f(|\ane|_r) + \sum_{r \in \bar{Q}^k \setminus Q^k} v_r \right) \\
		& = & \sum_{r \in \aopt \cap \ane} v_r f(|\ane|_r) + \sum_{r \in \aopt \setminus \ane} v_r, 
	\end{eqnarray*}
	where the above equality follows from the fact that $\bar{Q}^i \cap \bar{Q}^j = \emptyset$ for any $i \neq j$ which is due to the disjointness of $\aopt$. Using the definition of $U_i(\ane)$, we have
	$$\sum_{r \in \ane \setminus \aopt}\!\!\!\!\!\!\!\! v_r |\ane|_r f(|\ane|_r) + \!\!\!\!\!\!\!\! \sum_{r \in \ane \cap \aopt}\!\!\!\!\!\!\!\!\!v_r \left(|\ane|_r -1 \right) f(|\ane|_r)
	\geq\!\!\!\!\!\!\!\! \sum_{r \in \aopt \setminus \ane}\!\!\!\!\!\!\! v_r.
	$$
	Define $\gamma = \max_{j \leq n} (j-1)f(j)$.  Working with the above expression we have
	\begin{equation*} 
	\sum_{r \in \ane \setminus \aopt} v_r (\gamma + 1) + \sum_{r \in \ane\cap \aopt}v_r \gamma
	\geq \sum_{r \in \aopt \setminus \ane} v_r,
	\end{equation*}
	which gives us that 
	\begin{equation*} 
	(\gamma + 1) W(\ane) \geq W(\aopt)
	\end{equation*}
	which completes the lower bound.  
	
	We will now provide an accompanying upper bound on the price of stability.  To that end, consider a family of examples parameterized by a coefficient $j \in \{1, \dots, n \}$.  For each $j$, the game consists of $j$ agents and $(j+1)$-resources $\ree = \{r_0, r_1, \dots, r_j\}$ where the values of the resources are $v_{r_0} = 1$ and $v_{r_1} = \dots = v_{r_j} = f(j)-\varepsilon$ where $\varepsilon>0$ is an arbitrarily small constant, and the action set of each agent $i \in \{1, \dots, j\}$ is $\aee_i =\{r_0, r_i\}$. The unique equilibrium  is of the form $\ane = (r_0, \dots, r_0)$ as every agent selects resource $r_0$ and the total welfare is $W(\ane) = 1$.  The optimal allocation is of the form $\aopt = (r_0, r_2, \dots, r_j)$ which generates a total welfare of $W(\aopt)=1+(j-1)(f(j)-\varepsilon)$.  Performing a worst case analysis over $\varepsilon$ and $j$ gives \eqref{eq:pos2}.
The claim in \eqref{eq:981} follows from observing that $\fmc$ is always optimal, and returns a price of stability of $1$, thanks to Proposition~\ref{claim:mc-share}.
\qed

\bibliographystyle{plain}
\bibliography{scibib_modified} %

\begin{IEEEbiography}[{\includegraphics[width=1in,height=1.25in,clip,keepaspectratio]{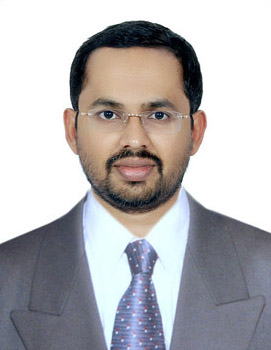}}]
	{Vinod Ramaswamy} received his B.Tech in electronics and communication engineering from Kerala University, India, in 2004 and his M.Tech. degree in electrical engineering, with specialization in communication systems, from Indian Institute of Technology, Madras, India, in 2006. He received his Ph.D. degree in Electrical Engineering from Texas A\&M University, College Station in 2013. From 2015 to 2016, he was a Postdoctoral Scholar at University of Colorado, Boulder. He is currently working in Qualcomm. His research interests include game theory, multi-agent systems and communication networks.
	
\end{IEEEbiography}
\vspace{-.5in}

\begin{IEEEbiography}[{\includegraphics[width=1in,height=1.25in,clip,keepaspectratio]{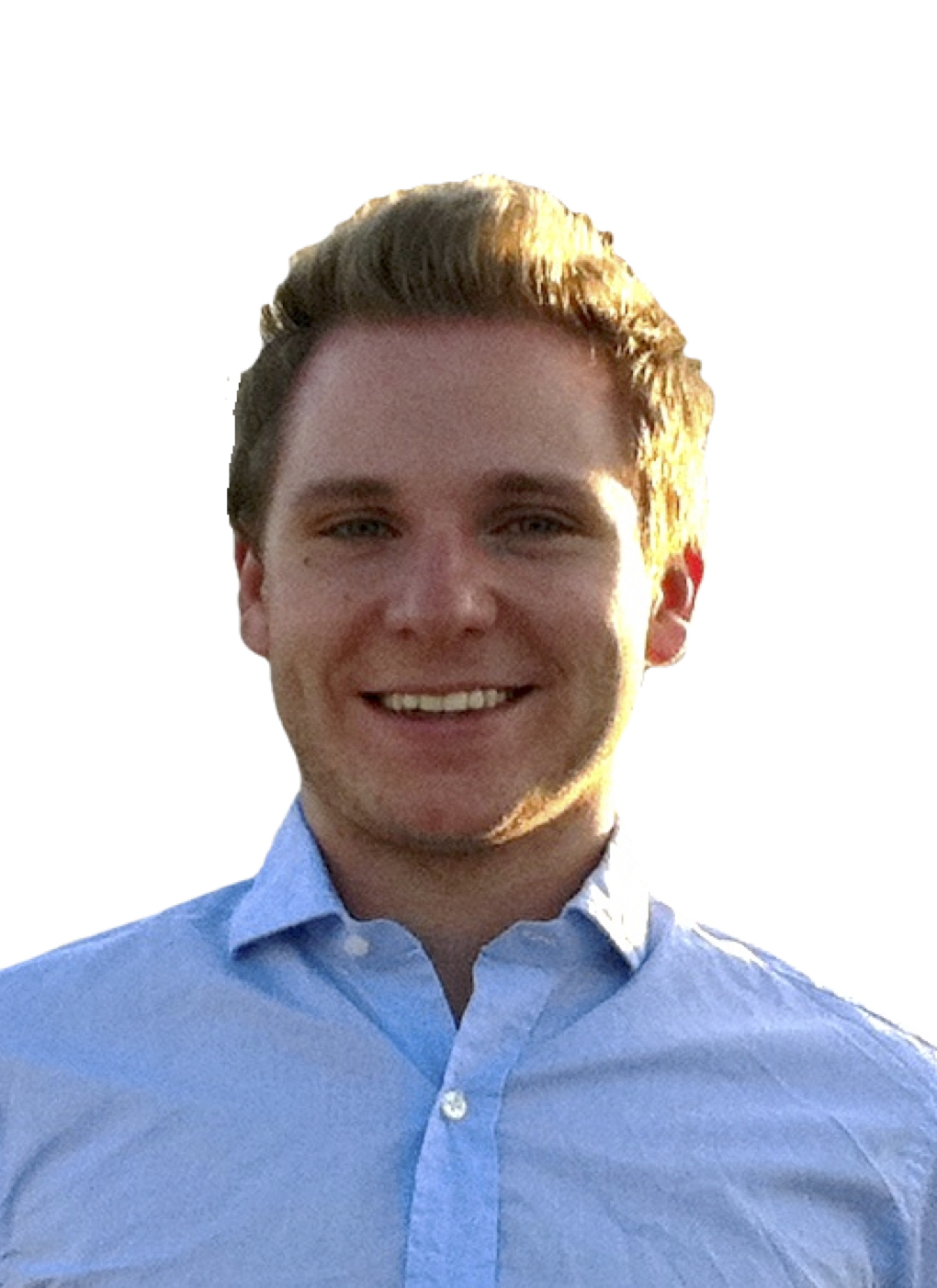}}]
	{Dario Paccagnan} is a Postdoctoral Fellow with the Mechanical Engineering Department and the Center for Control, Dynamical Systems and Computation, University of California, Santa Barbara.
In 2018, Dario obtained a Ph.D. degree from the Information Technology and Electrical Engineering Department, ETH Z\"{u}rich, Switzerland. He received a B.Sc. and M.Sc. in Aerospace Engineering in 2011 and 2014 from the University of Padova, Italy, and a M.Sc. in Mathematical Modelling and Computation from the Technical University of Denmark in 2014; all with Honors.
Dario was a visiting scholar at the University of California, Santa Barbara in 2017, and at Imperial College of London, in 2014.
His interests are at the interface between control theory and game theory, with a focus on the design of behavior-influencing mechanisms for socio-technical systems. Applications include multiagent systems and smart cities. Dr. Paccagnan was awarded the ETH medal, and is recipient of the SNSF fellowship for his work in Distributed Optimization and Game Design. 
\end{IEEEbiography}
\vspace{-.5in}

\begin{IEEEbiography}[{\includegraphics[width=1in,height=1.25in,clip,keepaspectratio]{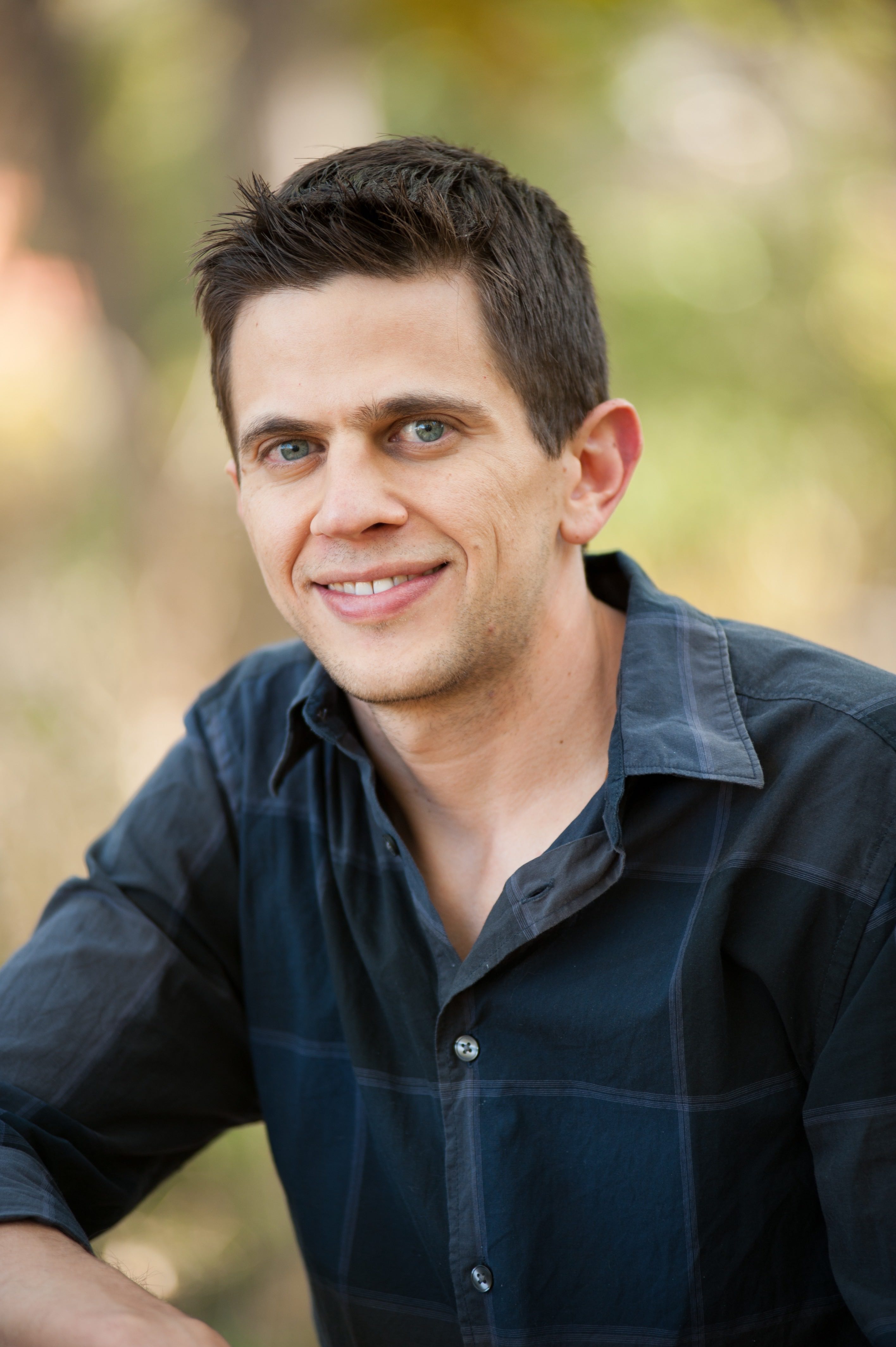}}]{Jason R. Marden} is an Associate Professor in the Department of Electrical and Computer Engineering at the University of California, Santa Barbara.  Jason received a BS in Mechanical Engineering in 2001 from UCLA, and a PhD in Mechanical Engineering in 2007, also from UCLA, where he was awarded the Outstanding Graduating PhD Student in Mechanical Engineering. After graduating from UCLA, he served as a junior fellow in the Social and Information Sciences Laboratory at the California Institute of Technology until 2010 when he joined the University of Colorado. In 2015, Jason joined the University of California, Santa Barbara.  Jason is a recipient of the NSF Career Award (2014), the ONR Young Investigator Award (2015), the AFOSR Young Investigator Award (2012), the American Automatic Control Council Donald P. Eckman Award (2012), and the SIAG/CST Best SICON Paper Prize (2015).   Jason's research interests focus on game theoretic methods for the control of distributed multi-agent systems.
\end{IEEEbiography}

\end{document}

%% file: figures/POAvsPOS_initial.tikz
% Minimal pgfplots version: 1.3
%  http://www.mathworks.com/matlabcentral/fileexchange/22022-matlab2tikz
\begin{tikzpicture}
%\pgfplotsset{every major tick/.append style={thick}}
\pgfplotsset{every major tick/.append style={black}}
\begin{axis}[%
width=3cm,
height=3cm,
at={(1.011111in,0.641667in)},
scale only axis,
axis equal image,
xmin=0,
xmax=1.05,
ymin=0,
ymax=1.05,
xlabel={{\rm PoA}},
y label style={at={(axis description cs:0.2,.45)},anchor=south},
ylabel={{\rm PoS}},
ylabel style = {align=right},
legend style={draw=none},
]
%\addlegendimage{only marks, color=green, mark = triangle, mark size = 4pt, mark options={solid}, line width = 1pt}
\addplot[mark = o, mark size = 2pt, mark options={solid}, line width = 1pt]
  table[row sep=crcr]
  {0.5	0.5\\};
%\addlegendentry{~$\fes$};
%
\addplot[mark = x, mark size = 3pt, mark options={solid}, line width = 1pt]
  table[row sep=crcr]
  {0.6321205588	0.6321205588\\};
%\addlegendentry{~$\fgair$};
%
\addplot[mark = triangle, mark size = 3pt, mark options={solid}, line width = 1pt]
  table[row sep=crcr]
  {0.5	1\\};
\node[anchor=west] at (axis cs:.3, 0.88){{$\fmc$}};
\node[anchor=west] at (axis cs:.3, 0.38){{$\fes$}};
\node[anchor=west] at (axis cs:.6, 0.5){{$\fgair$}};
\end{axis}
%
%
%\node[anchor=west] at (axis cs:.3, 0.9){{$\fgair$}};
%
%%
%\draw[->,thick] (axis cs:.45, 0.77) -- (axis cs:.54, 0.82);
%\draw[->,thick] (axis cs:.45, 0.63) -- (axis cs:.6, 0.63);
%\draw[->,thick] (axis cs:.45, 0.935) -- (axis cs:.49, 0.98);
%\draw[->,thick] (axis cs:.32, 0.87) -- (axis cs:.32, 0.98);
%\node[anchor=west] at (axis cs:.24, 0.77){{\small $Z(\alpha,n)$}};
%\node[anchor=west] at (axis cs:.3, 0.63){{\section$\fgair$}};
%\node[anchor=west] at (axis cs:.3, 0.9){{$\fmc$}};
%\node[mark size=2.5pt, anchor= center, line width=1pt] at (axis cs: 0.632120558828558,0.632120558828558) {\pgfuseplotmark{x}};
%\node[mark size=2pt, anchor= center, line width=1pt] at (axis cs: 0.5,0.5) {\pgfuseplotmark{o}};
%\node[mark size=2pt, anchor= center, line width=1pt] at (axis cs: 0.5,1) {\pgfuseplotmark{triangle}};
%\addlegendentry{a},
%\node[anchor=west] at (axis cs:0.38, 0.57){{\tiny${\rm PoA} = 1-1/e,1-\frac{1}{e})$}};
\end{tikzpicture}%

%% file: figures/POAvsPOS.tikz
% Minimal pgfplots version: 1.3
%  http://www.mathworks.com/matlabcentral/fileexchange/22022-matlab2tikz
\begin{tikzpicture}
%\pgfplotsset{every major tick/.append style={thick}}
\pgfplotsset{every major tick/.append style={black}}
\begin{axis}[%
width=\figurewidth,
height=\figureheight,
at={(1.011111in,0.641667in)},
scale only axis,
axis equal image,
xmin=0,
xmax=1,
ymin=0,
ymax=1,
xlabel={Worst case performance ({\rm PoA})},
ylabel={Best case performance ({\rm PoS})},
%ylabel style = {align=left},
%ylabel={Price of anarchy},
%title style={font=\bfseries},
%title={Optimal PoA},
%ticklabel style = {font=\large},
%ymajorgrids
%axis x line*=bottom,
%axis y line*=left
]

\draw[->,thick] (axis cs:.45, 0.77) -- (axis cs:.54, 0.82);
\draw[->,thick] (axis cs:.45, 0.63) -- (axis cs:.6, 0.63);
\draw[->,thick] (axis cs:.45, 0.935) -- (axis cs:.49, 0.98);
%\draw[->,thick] (axis cs:.32, 0.87) -- (axis cs:.32, 0.98);
\node[anchor=west] at (axis cs:.24, 0.77){{\small $Z(\alpha,n)$}};
\node[anchor=west] at (axis cs:.3, 0.63){{$\fgair$}};
\node[anchor=west] at (axis cs:.3, 0.9){{$\fmc$}};
\addplot[area legend, solid, color=gray!50, fill=gray!50,forget plot]
table[row sep=crcr] {%
x	y\\
0	0\\
1	1\\
1	0\\
}--cycle;

\addplot[area legend, color=red!40, fill=red!40, forget plot]
table[row sep=crcr] {%
x	y\\
0.5	1\\
0.501	0.996023856858847\\
0.502	0.992094861660079\\
0.503	0.988212180746562\\
0.504	0.984375\\
0.505	0.980582524271845\\
0.506	0.976833976833977\\
0.507	0.973128598848369\\
0.508	0.969465648854962\\
0.509	0.96584440227704\\
0.51	0.962264150943396\\
0.511	0.958724202626642\\
0.512	0.955223880597015\\
0.513	0.951762523191095\\
0.514	0.948339483394834\\
0.515	0.944954128440367\\
0.516	0.941605839416058\\
0.517	0.938294010889292\\
0.518	0.935018050541516\\
0.519	0.931777378815081\\
0.52	0.928571428571428\\
0.521	0.925399644760213\\
0.522	0.92226148409894\\
0.523	0.919156414762742\\
0.524	0.916083916083916\\
0.525	0.91304347826087\\
0.526	0.910034602076124\\
0.527	0.907056798623064\\
0.528	0.904109589041096\\
0.529	0.901192504258944\\
0.53	0.898305084745763\\
0.531	0.895446880269814\\
0.532	0.89261744966443\\
0.533	0.889816360601002\\
0.534	0.887043189368771\\
0.535	0.884297520661157\\
0.536	0.881578947368421\\
0.537	0.878887070376432\\
0.538	0.876221498371335\\
0.539	0.873581847649919\\
0.54	0.870967741935484\\
0.541	0.868378812199037\\
0.542	0.865814696485623\\
0.543	0.863275039745628\\
0.544	0.860759493670886\\
0.545	0.858267716535433\\
0.546	0.855799373040752\\
0.547	0.853354134165366\\
0.548	0.850931677018633\\
0.549	0.848531684698609\\
0.55	0.846153846153846\\
0.551	0.843797856049005\\
0.552	0.841463414634146\\
0.553	0.839150227617602\\
0.554	0.836858006042296\\
0.555	0.834586466165413\\
0.556	0.832335329341317\\
0.557	0.8301043219076\\
0.558	0.827893175074184\\
0.559	0.825701624815362\\
0.56	0.823529411764706\\
0.561	0.821376281112738\\
0.562	0.819241982507289\\
0.563	0.817126269956459\\
0.564	0.815028901734104\\
0.565	0.81294964028777\\
0.566	0.810888252148997\\
0.567	0.808844507845934\\
0.568	0.806818181818182\\
0.569	0.804809052333805\\
0.57	0.802816901408451\\
0.571	0.800841514726508\\
0.572	0.798882681564246\\
0.573	0.796940194714882\\
0.574	0.795013850415513\\
0.575	0.793103448275862\\
0.576	0.791208791208791\\
0.577	0.789329685362517\\
0.578	0.787465940054496\\
0.579	0.78561736770692\\
0.58	0.783783783783784\\
0.581	0.781965006729475\\
0.582	0.780160857908847\\
0.583	0.778371161548732\\
0.584	0.776595744680851\\
0.585	0.774834437086093\\
0.586	0.773087071240106\\
0.587	0.771353482260184\\
0.588	0.769633507853403\\
0.589	0.767926988265971\\
0.59	0.766233766233766\\
0.591	0.764553686934023\\
0.592	0.762886597938144\\
0.593	0.761232349165597\\
0.594	0.759590792838875\\
0.595	0.75796178343949\\
0.596	0.756345177664975\\
0.597	0.754740834386852\\
0.598	0.753148614609572\\
0.599	0.751568381430364\\
0.6	0.75\\
0.601	0.748443337484433\\
0.602	0.746898263027295\\
0.603	0.745364647713226\\
0.604	0.74384236453202\\
0.605	0.742331288343558\\
0.606	0.740831295843521\\
0.607	0.739342265529842\\
0.608	0.737864077669903\\
0.609	0.73639661426844\\
0.61	0.734939759036145\\
0.611	0.733493397358944\\
0.612	0.732057416267943\\
0.613	0.730631704410012\\
0.614	0.729216152019002\\
0.615	0.727810650887574\\
0.616	0.726415094339623\\
0.617	0.72502937720329\\
0.618	0.723653395784543\\
0.619	0.722287047841307\\
0.62	0.72093023255814\\
0.621	0.719582850521437\\
0.622	0.71824480369515\\
0.623	0.716915995397008\\
0.624	0.715596330275229\\
0.625	0.714285714285714\\
0.626	0.706546275395034\\
0.627	0.698996655518395\\
0.628	0.691629955947137\\
0.629	0.684439608269859\\
0.63	0.677419354838709\\
0.631	0.670563230605739\\
0.632120558828558	0.632120558828558\\
1	1\\
}--cycle;
\addplot[color=black,solid,line width=1pt]
  table[row sep=crcr]{%
0   0.998\\
0.5	0.998\\
0.501	0.996023856858847\\
0.502	0.992094861660079\\
0.503	0.988212180746562\\
0.504	0.984375\\
0.505	0.980582524271845\\
0.506	0.976833976833977\\
0.507	0.973128598848369\\
0.508	0.969465648854962\\
0.509	0.96584440227704\\
0.51	0.962264150943396\\
0.511	0.958724202626642\\
0.512	0.955223880597015\\
0.513	0.951762523191095\\
0.514	0.948339483394834\\
0.515	0.944954128440367\\
0.516	0.941605839416058\\
0.517	0.938294010889292\\
0.518	0.935018050541516\\
0.519	0.931777378815081\\
0.52	0.928571428571428\\
0.521	0.925399644760213\\
0.522	0.92226148409894\\
0.523	0.919156414762742\\
0.524	0.916083916083916\\
0.525	0.91304347826087\\
0.526	0.910034602076124\\
0.527	0.907056798623064\\
0.528	0.904109589041096\\
0.529	0.901192504258944\\
0.53	0.898305084745763\\
0.531	0.895446880269814\\
0.532	0.89261744966443\\
0.533	0.889816360601002\\
0.534	0.887043189368771\\
0.535	0.884297520661157\\
0.536	0.881578947368421\\
0.537	0.878887070376432\\
0.538	0.876221498371335\\
0.539	0.873581847649919\\
0.54	0.870967741935484\\
0.541	0.868378812199037\\
0.542	0.865814696485623\\
0.543	0.863275039745628\\
0.544	0.860759493670886\\
0.545	0.858267716535433\\
0.546	0.855799373040752\\
0.547	0.853354134165366\\
0.548	0.850931677018633\\
0.549	0.848531684698609\\
0.55	0.846153846153846\\
0.551	0.843797856049005\\
0.552	0.841463414634146\\
0.553	0.839150227617602\\
0.554	0.836858006042296\\
0.555	0.834586466165413\\
0.556	0.832335329341317\\
0.557	0.8301043219076\\
0.558	0.827893175074184\\
0.559	0.825701624815362\\
0.56	0.823529411764706\\
0.561	0.821376281112738\\
0.562	0.819241982507289\\
0.563	0.817126269956459\\
0.564	0.815028901734104\\
0.565	0.81294964028777\\
0.566	0.810888252148997\\
0.567	0.808844507845934\\
0.568	0.806818181818182\\
0.569	0.804809052333805\\
0.57	0.802816901408451\\
0.571	0.800841514726508\\
0.572	0.798882681564246\\
0.573	0.796940194714882\\
0.574	0.795013850415513\\
0.575	0.793103448275862\\
0.576	0.791208791208791\\
0.577	0.789329685362517\\
0.578	0.787465940054496\\
0.579	0.78561736770692\\
0.58	0.783783783783784\\
0.581	0.781965006729475\\
0.582	0.780160857908847\\
0.583	0.778371161548732\\
0.584	0.776595744680851\\
0.585	0.774834437086093\\
0.586	0.773087071240106\\
0.587	0.771353482260184\\
0.588	0.769633507853403\\
0.589	0.767926988265971\\
0.59	0.766233766233766\\
0.591	0.764553686934023\\
0.592	0.762886597938144\\
0.593	0.761232349165597\\
0.594	0.759590792838875\\
0.595	0.75796178343949\\
0.596	0.756345177664975\\
0.597	0.754740834386852\\
0.598	0.753148614609572\\
0.599	0.751568381430364\\
0.6	0.75\\
0.601	0.748443337484433\\
0.602	0.746898263027295\\
0.603	0.745364647713226\\
0.604	0.74384236453202\\
0.605	0.742331288343558\\
0.606	0.740831295843521\\
0.607	0.739342265529842\\
0.608	0.737864077669903\\
0.609	0.73639661426844\\
0.61	0.734939759036145\\
0.611	0.733493397358944\\
0.612	0.732057416267943\\
0.613	0.730631704410012\\
0.614	0.729216152019002\\
0.615	0.727810650887574\\
0.616	0.726415094339623\\
0.617	0.72502937720329\\
0.618	0.723653395784543\\
0.619	0.722287047841307\\
0.62	0.72093023255814\\
0.621	0.719582850521437\\
0.622	0.71824480369515\\
0.623	0.716915995397008\\
0.624	0.715596330275229\\
0.625	0.714285714285714\\
0.626	0.706546275395034\\
0.627	0.698996655518395\\
0.628	0.691629955947137\\
0.629	0.684439608269859\\
0.63	0.677419354838709\\
0.631	0.670563230605739\\
0.632120558828558	0.632120558828558\\
};
%%
%\addplot[color=black,dashed,line width=0.5pt]
%  table[row sep=crcr]{%
%0.5 1\\
%0.5 0\\
%};
\node[anchor=west] at (axis cs:.53, 0.92){\textcolor{red}{\small not achievable}};
%\node[anchor=west] at (axis cs:.56, 0.86){\textcolor{red}{Thm.~\ref{thm:tradeoff}}};
\node[mark size=3pt, anchor= center, line width=1.5pt] at (axis cs: 0.632120558828558,0.632120558828558) {\pgfuseplotmark{x}};
%\node[mark size=3pt, anchor= center, line width=1.5pt] at (axis cs: 0.45,1) {\pgfuseplotmark{x}};
\node[anchor=west] at (axis cs:0.58, 0.57){{\small$(1-\frac{1}{e},1-\frac{1}{e})$}};
\end{axis}
\end{tikzpicture}%